\documentclass[a4paper,UKenglish]{lipics-v2016}
\usepackage{graphicx}
\usepackage{algorithm}

\usepackage{algpseudocode}
\usepackage{multicol}
\usepackage{amsmath}
\usepackage{amsthm}

\usepackage{amssymb}
\usepackage{float}
\usepackage{relsize}
\allowdisplaybreaks
\usepackage[font={it}]{caption}
\usepackage{verbatim}
\usepackage{enumerate}
\newtheorem*{theorem*}{Theorem}
\newtheorem{Fact}{Fact}
\newtheorem*{lemma*}{Lemma}
\newtheorem{observation}[theorem]{Observation}

\usepackage{pifont}
\usepackage{amsmath,bm}
\algtext*{EndFunction}
\algtext*{EndIf}
	
\newcommand*{\TWELVEPAGER}{}
\title{Give Me Some Slack: Efficient Network Measurements}

\author[1]{Ran Ben Basat}
\author[2]{Gil Einziger}
\author[1]{Roy Friedman}

\affil[1]{Department of Computer Science,	Technion\\
	\{sran,roy\}@cs.technion.ac.il}
\affil[2]{Nokia Bell Labs\\
	{gil.einziger@nokia.com}}
\authorrunning{Ran Ben-Basat, Gil Einziger, and Roy Friedman} 

\Copyright{Ran Ben-Basat, Gil Einziger, Roy Friedman}

\date{}
\begin{document}
\maketitle


\newcommand{\altLogBase}{\ensuremath{(1+\epsilon/3)}}
\newcommand{\remRoundDown}[1]{\parentheses{#1}_{\underline{\overline{\Downarrow}}}}

\renewcommand{\SS}{{\sc Slack Summing}}
\newcommand{\EXACT}{$(W,\tau)${\sc -Exact Summing}}
\newcommand{\MULT}{$(W,\tau,\epsilon)${\sc -Multiplicative Summing}}
\newcommand{\ADDI}{$(W,\tau,\epsilon)${\sc -Additive Summing}}
\newcommand{\OUT}{{\sc Output}}
\newcommand{\cSet}{0\le c< \wt}
\newcommand{\alg}{\ensuremath{\mathbb A}}
\newcommand{\poly}{\text{poly}}
\newcommand{\brackets}[1]{\left[#1\right]}

\newcommand{\resizeIfTWELVEPAGER}[1]
{
\ifdefined \TWELVEPAGER
\mathsmaller{#1}
\else
#1
\fi
}

\newcommand{\algsize}{\small}

\newcommand{\eps}{\epsilon}
\newcommand{\set}[1]{\left\{#1\right\}}
\newcommand{\ceil}[1]{ \left\lceil{#1}\right\rceil}
\newcommand{\floor}[1]{ \left\lfloor{#1}\right\rfloor}
\newcommand{\logp}[1]{\log\parentheses{#1}}
\newcommand{\lnp}[1]{\ln\parentheses{#1}}
\newcommand{\Omegap}[1]{\Omega\parentheses{#1}}
\newcommand{\clog}[1]{ \ceil{\log{#1}}}
\newcommand{\clogp}[1]{ \ceil{\logp{#1}} }
\newcommand{\flog}[1]{ \floor{\log{#1}}}
\newcommand{\flogp}[1]{ \floor{\logp{#1}}}
\newcommand{\parentheses}[1]{ \left({#1}\right)}
\newcommand{\abs}[1]{ \left|{#1}\right|}

\newcommand{\cdotpa}[1]{\cdot\parentheses{#1}}
\newcommand{\inc}[1]{$#1 = #1 + 1$}
\newcommand{\dec}[1]{$#1 = #1 - 1$}
\newcommand{\range}[2][0]{#1,1,\ldots,#2}
\newcommand{\orange}[1]{\set{1,\ldots,#1}}
\newcommand{\frange}[1]{\set{\range{#1}}}
\newcommand{\xrange}[1]{\frange{#1-1}}
\newcommand{\oneOverE}{ \eps^{-1} }
\newcommand{\oneOverT}{ \tau^{-1} }
\newcommand{\smallMultError}{(1+o(1))}
\newcommand{\lowerbound}{\max \set{\log W ,\frac{1}{2\epsilon+W^{-1}}}}
\newcommand{\smallEpsLowerbound}{\window\logp{\frac{1}{\weps}}}
\newcommand{\smallEpsMemoryTheta}{$\Theta\parentheses{\smallEpsMemoryConsumption}$}
\newcommand{\smallEpsMemoryConsumption}{W\cdot\logp{\frac{1}{\weps}}}

\newcommand{\RR}{\frange{R}}
\renewcommand{\RR}{[R]}

\newcommand{\largeEpsRestriction}{For any \largeEps{},}
\newcommand{\largeEps}{$\eps^{-1} \le 2W\left(1-\frac{1}{\logw}\right)$}
\newcommand{\smallEpsRestriction}{For any \smallEps{},}
\newcommand{\smallEps}{$\eps^{-1}>2W\left(1-\frac{1}{\logw}\right)=2\window(1-o(1))$}
\newcommand{\bc}{{\sc Basic-Counting}}
\newcommand{\bs}{{\sc Basic-Summing}}
\newcommand{\gs}{{\sc General-Summing}}
\newcommand{\cdp}{{\sc Count-Distinct}}
\newcommand{\maxim}{{\sc Maximum}}
\newcommand{\std}{{\sc Standard-Deviation}}
\newcommand{\windowcounting}{ {\sc $(W,\epsilon)$-Window-Counting}}

\newcommand{\query}{{\sc Output}}
\newcommand{\add}  [1][] {{\sc Update}$(#1)$}

\newcommand{\window}{W}
\newcommand{\logw}{\log \window}
\newcommand{\flogw}{\floor{\log \window}}
\newcommand{\weps}{\window\epsilon}
\newcommand{\wt}{\window\tau}
\newcommand{\logweps}{\logp{\weps}}
\newcommand{\logwt}{\logp{\wt}}
\newcommand{\bitarray}{b}
\newcommand{\currentBlockIndex}{i}
\newcommand{\currentBlock}{\bitarray_{\currentBlockIndex}}
\newcommand{\remainder}{y}
\newcommand{\numBlocks}{\oneOverT}
\newcommand{\sumOfBits}{B}
\newcommand{\blockSize}{\frac{\window}{\numBlocks}}
\newcommand{\iblockSize}{\frac{\numBlocks}{\window}}
\newcommand{\threshold}{\blockSize}
\newcommand{\halfBlock}{\frac{\window}{2\numBlocks}}
\newcommand{\blockOffset}{c}
\newcommand{\inputVariable}{x}

\newcommand{\bcTableColumnWidth}{1.5cm}
\newcommand{\bsTableColumnWidth}{1.7cm}
\newcommand{\bsExtendedTableColumnWidth}{3cm}
\newcommand{\bcExtendedTableColumnWidth}{2.8cm}
\newcommand{\bcNarrowTableColumnWidth}{1.5cm}
\newcommand{\bsNarrowTableColumnWidth}{1.5cm}
\newcommand{\bsWorstCaseTableColumnWidth}{2cm}

\newcommand{\bsrange}{ R }
\newcommand{\bsReminderPercisionParameter}{ \gamma }
\newcommand{\bsest}{ \widehat{\bssum}}
\newcommand{\bssum}{ S^W }
\newcommand{\bsFracInput}{ \inputVariable' }
\newcommand{\bserror}{ \bsrange\window\epsilon }
\newcommand{\bsfractionbits}{ \frac{\bsReminderPercisionParameter}{\epsilon} }
\newcommand{\bsReminderFractionBits}{ \upsilon}
\newcommand{\bsAnalysisTarget}{ \bssum}
\newcommand{\bsAnalysisEstimator}{ \widehat{\bsAnalysisTarget}}
\newcommand{\bsAnalysisError}{ \bsAnalysisEstimator - \bsAnalysisTarget}
\newcommand{\bsRoundingError}{ \xi}


\newcommand{\neps}{\ensuremath{\winSize\eps}}
\newcommand{\Neps}{\ensuremath{\maxWinSize\eps}}
\newcommand{\logn}{\ensuremath{\log\winSize}}
\newcommand{\logN}{\ensuremath{\log\maxWinSize}}
\newcommand{\logneps}{\ensuremath{\logp\neps}}
\newcommand{\logNeps}{\ensuremath{\logp\Neps}}
\newcommand{\oneOverEps}{\ensuremath{\frac{1}{\eps}}}
\newcommand{\winSize}{\ensuremath{n}}
\newcommand{\maxWinSize}{\ensuremath{N}}
\newcommand{\curTime}{\ensuremath{t}}
\newcommand\Tau{\mathrm{T}}
\newcommand{\offset}{\ensuremath{\mathit{offset}}}
\newcommand{\roundedOOE}{k}
\newcommand{\numLargeBlocks}{\frac{\roundedOOE}{4}}
\newcommand{\numSmallBlocks}{\frac{\roundedOOE}{2}}

\newcommand{\remove}{{\sc Remove()}}
\newcommand{\merge}[1]{{\sc Merge(#1)}}
\newcommand{\counting}{{\sc Counting}}
\newcommand{\summing}{{\sc Summing}}
\newcommand{\freq}{{\sc Frequency Estimation}}
\newcommand{\SSS}{CSS}
\newcommand{\SpaceS}{Space Saving}
\newcommand{\CSS}{Compact \SpaceS{}}
\newcommand{\WCSS}{Window \CSS{}}
\newcommand{\SSSInstance}{y}
\newcommand{\queueOfOverflows}{b}
\newcommand{\sumOfOverflows}{B}
\newcommand{\IDArray}{O}

\newcommand{\frameOffset}{M}
\newcommand{\overflowIndicator}{u}
\newcommand{\xFrequency}{f_x}
\newcommand{\xFrequencyEstimator}{\widehat{f_x}}
\newcommand{\xWindowFrequency}{f^W_x}
\newcommand{\xWindowFrequencyEstimator}{\widehat{\xWindowFrequency}}

\newcommand{\idIndex}{ID-Index}
\newcommand{\valIndex}{Value-Index}
\newcommand{\cArray}{Counter-Array}
\newcommand{\reverseMapping}{$\mathbf R$}


\newcommand{\heavyHitters}
{\ensuremath{\begin{aligned}(\window,\theta,\eps)\end{aligned}}-~\textsc{Heavy Hitters}}
\newcommand{\streamcounting}{$\epsilon$\textsc{-Counting}}
\newcommand{\paremeterizedStreamcounting}[1]{#1\textsc{-Counting}}
\newcommand{\probabilisticStreamcounting}{$(\epsilon, \delta)$\textsc{-Counting}}
\newcommand{\probabilisticWindowcounting}{$(\window,\epsilon, \delta)$\textsc{-Window Counting}}

\begin{abstract}

Many networking applications require timely access to recent network measurements, which can be captured using a sliding window model.
Maintaining such measurements is a challenging task due to the fast line speed and scarcity of fast memory in routers.
In this work, we study the impact of allowing \emph{slack} in the window size on the asymptotic requirements of sliding window problems.
That is, the algorithm can dynamically adjust the window size between $W$ and $W(1+\tau)$ where $\tau$ is a small positive parameter.
We demonstrate this model's attractiveness by showing that it enables efficient algorithms to problems such as \maxim{} and \gs{} that require $\Omega(W)$ bits even for constant factor approximations in the exact sliding window model.
Additionally, for problems that admit sub-linear approximation algorithms such as \bs{} and \cdp, the slack model enables a further asymptotic improvement.

The main focus of the paper is on the widely studied \bs{} problem of computing the sum of the last $W$ integers from $\set{0,1\ldots,R}$ in a stream.
While it is known that $\Omega(W\log{R})$ bits are needed in the exact window model, we show that approximate windows allow an \emph{exponential} space reduction for constant $\tau$.

Specifically, for $\tau=\Theta(1)$, we present a space lower bound of $\Omega(\log(RW))$ bits.
Additionally, we show an $\Omega(\logp{W/\epsilon})$ lower bound for $RW\epsilon$ additive approximations and a $\Omega(\logp{W/\epsilon}+\log\log{R})$ bits lower bound for $(1+\epsilon)$ multiplicative approximations.
Our work is the first to study this problem in the exact and additive approximation settings.
For all settings, we provide memory optimal algorithms that operate in \emph{worst case} constant time.
This strictly improves on the work of~\cite{DatarGIM02} for $(1+\epsilon)$-multiplicative approximation that requires $O(\epsilon^{-1} \logp{RW}\log\logp{RW})$ space and performs updates in $O(\logp{RW})$ worst case time.
Finally, we show asymptotic improvements for the \cdp{}, \gs{} and \maxim{} problems.
\end{abstract}

\section{Introduction}
Network algorithms in diverse areas such as traffic engineering, load balancing and quality of service ~\cite{CONGA,TrafficEngeneering,IntrusionDetection2,ApproximateFairness,7218487} rely on timely link measurements.
In such applications recent data is often more relevant than older data, motivating the notions of \emph{aging} and \emph{sliding window}~\cite{WCSS,SlidingHLL,TinyLFU,SlidingWindowBF,SlidingBloomFilter}.
For example, a sudden decrease in the average packet size on a link may indicate a SYN attack~\cite{IntrusionDetection}.
Additionally, a load balancer may benefit from knowing the current utilization of a link to avoid congestion~\cite{CONGA}.

While conceptually simple, conveying the necessary information to network algorithms is a difficult challenge due to current memory technology limitations.
Specifically, DRAM memory is abundant but too slow to cope with the line rate while SRAM memory is fast enough but has a limited capacity~\cite{dleftCBF,CounterTree,CounterArchitecture}.
Online decisions are therefore realized through space efficient data structures~\cite{dimsum,RAP,ICE-Buckets,TinyTable,RHHH,Brick,CEDAR,CASE} that store measurement statistics in a concise manner.
For example,~\cite{ICE-Buckets,CEDAR} utilize probabilistic counters that only require $O(\log\log N)$ bits to approximately represent numbers up to $N$.
Others conserve space using variable sized counter encoding~\cite{TinyTable,Brick} and monitoring only the frequent elements~\cite{WCSS}.

\ifdefined \TWELVEPAGER
\else
\bc{} is one of the textbook problems of such approximated sliding window stream processing~\cite{DatarGIM02}.
In this problem, one is required to keep track of the number of $1$'s in the last $W$ elements of a stream of binary bits.
In their seminal work,~\cite{DatarGIM02} presents a \emph{$(1+\eps)$-multiplicative approximation} algorithm for this problem using $O\left(\frac{1}{\eps}\log^2\window\eps\right)$ bits, as well as a matching lower bound.
Their solution works with amortized $O(1)$ time, but its worst-case time complexity is $O(\logw)$.
Additionally, the work of~\cite{SWATPAPER} provides a \emph{$RW\epsilon$-additive approximation} to the problem using $O\left(\frac{1}{\eps} + \log\window\eps\right)$ bits as well as a matching lower bound.
\fi

\ifdefined \TWELVEPAGER
\bs{} is one of the most basic textbook examples of such approximated sliding window stream processing problems~\cite{DatarGIM02}.
In this problem, one is required to keep track of the sum of the last $\window$ elements, when all elements are non-negative integers in the range $\frange\bsrange$.
The work in~\cite{DatarGIM02} provides a $(1+\eps)$-multiplicative approximation of this problem using $O\parentheses{\frac{1}{\eps}\cdot \parentheses{\log^2\window +\log\bsrange\cdotpa{\logw+\log\log\bsrange}}}$ bits.
The amortized time complexity is $O(\frac{\log\bsrange}{\logw})$ and the worst case is $O(\logw+\log\bsrange)$.
\else
A more practical variant is \bs{} in which the goal is to maintain the sum of the last $\window$ elements.
When all elements are non-negative integers in the range $[\bsrange+1]=\frange\bsrange$, the work in~\cite{DatarGIM02} naturally extends to provide a $(1+\eps)$-multiplicative approximation of this problem using $O\parentheses{\frac{1}{\eps}\cdot \parentheses{\log^2\window +\log\bsrange\cdotpa{\logw+\log\log\bsrange}}}$ bits.
The amortized time complexity becomes $O(\frac{\log\bsrange}{\logw})$ and the worst case is $O(\logw+\log\bsrange)$.
\fi
In contrast, we previously showed an $RW\epsilon$-additive approximation with $\Theta\left(\frac{1}{\eps} + \log\window\eps\right)$ bits~\cite{SWATPAPER}.

Sliding window counters (approximated or accurate) require asymptotically more space than plain stream counters.
Such window counters are prohibitively large for networking devices which already optimize the space consumption of plain counters.

This paper explores the concept of \emph{slack}, or approximated sliding window, bridging this gap.
Figure~\ref{fig:slackyCounting} illustrates a ``window'' in this model.
Here, each query may select a \emph{$\tau$-slack window} whose size is between $W$ (the green elements) and $W(1+\tau)$ (the green plus yellow elements).
The goal is to compute the sum with respect to this chosen window.


\begin{figure}[t]
	\centering
	\includegraphics[width=\linewidth]{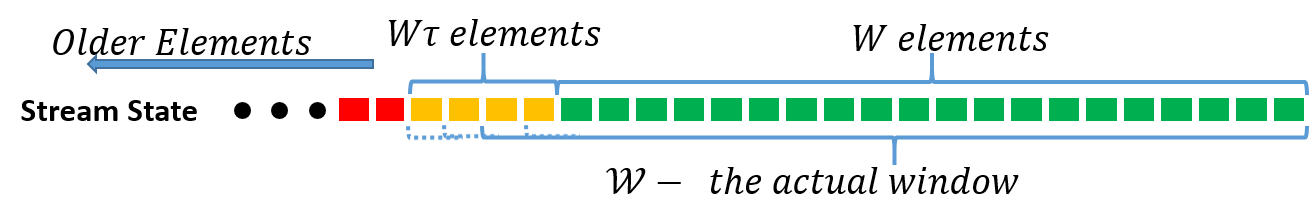}
	\caption{We need to answer each query with respect to a $\tau$-slack window that must include the last $W$ items, but may or may not consider a suffix of the previous $\wt$ elements.}
	\label{fig:slackyCounting}
\end{figure}

Slack windows were also considered in previous works~\cite{DatarGIM02,SlidingBloomFilter} and we call the problem of maintaining the sum over a slack window \SS{}.
\ifdefined \TWELVEPAGER
Datar et al.~\cite{DatarGIM02} showed that
\else
Datar et al.~\cite{DatarGIM02} showed that while computing a $(1+\eps)$-multiplicative approximation for \bc{} requires $\Omega(\oneOverE\log^2\weps)$ memory bits.
When $\tau=1$, they can compute such an approximation using only $O(\oneOverE\logweps\log\logw)$ bits.
For the \bs{} problem, a
\fi
constant slack reduces the required memory from $O({\frac{1}{\eps}\cdot \parentheses{\log^2\window +\log\bsrange\cdotpa{\logw+\log\log\bsrange}}})$ to $O(\oneOverE\log(RW)\log\log(RW))$.
\ifdefined \TWELVEPAGER
For $\tau$-slack windows they provide a $(1+\eps)$-multiplicative approximation using $O(\oneOverE\allowbreak\log(RW)(\log\log(RW) + \log\oneOverT))$ bits.
\else
For $\tau$-slack windows, their space consumption is $O(\oneOverE\logweps(\log\logw + \log\oneOverT))$ and $O(\oneOverE\log(RW)(\log\log(RW) + \log\oneOverT))$ for providing a $(1+\eps)$-multiplicative approximation for \bc{} and \bs{}, respectively.
\fi

\paragraph*{Our Contributions}

This paper studies the space and time complexity reductions that can be attained by allowing \emph{slack} -- an error in the window size.
Our results demonstrate exponentially smaller and asymptotically faster data structures compared to various problems over exact windows.
We start with deriving lower bounds for three variants of the \bs{} problem -- when computing an exact sum over a slack window, or when combined with an additive and a multiplicative error in the sum.
We present algorithms that are based on dividing the stream into $\wt$-sized blocks. 
Our algorithms sum the elements within each block and represent each block’s sum in a cyclic array of size $\oneOverT$. 
We use multiple compression techniques during different stages to drive down the space complexity. 
The resulting algorithms are space optimal, substantially simpler than previous work, and reduce update time to $O(1)$.

For exact \SS{}, we present a lower bound of $\Omega(\tau^{-1}\log(RW\tau))$ bits.
For $(1+\epsilon)$ multiplicative approximations we prove an $\Omega\big(\log (W/\eps)\allowbreak+ \oneOverT\parentheses{\logp{\tau/\eps} + \log\logp{RW}}\big)$ bits bound when $\tau=\Omegap{1\over \log{RW}}$.
We show that $\Omega(\tau^{-1}\log\floor{1+\tau/\epsilon}+\logp{W/\eps})$ bits are required for $RW\epsilon$ additive approximations.

Next, we introduce algorithms for the \SS{} problem, which asymptotically reduce the required memory compared to the sliding window model.
For the exact and additive error versions of the problem, we provide memory optimal algorithms.
In the multiplicative error setting, we provide an $O\big(\oneOverT\parentheses{\log{\oneOverE} + \log\logp {RW\tau}}+\log(RW)\big)$ space algorithm.
This is asymptotically optimal when $\tau=\Omega(\log^{-1}{W})$ and $R=\poly(W)$.
It also asymptotically improves~\cite{DatarGIM02} when $\oneOverT=o(\oneOverE\logp{RW})$.
We further provide an asymptotically optimal solution for constant $\tau$, even when $R=W^{\omega(1)}$.
All our algorithms are deterministic and operate in worst case constant time. In contrast, the algorithm of~\cite{DatarGIM02} works in $O(\log RW)$ worst case~time.

To exemplify our results, consider monitoring the average bandwidth (in bytes per second) passed through a router in a $24$ hours window, i.e., $W\triangleq 86400$ seconds.
Assuming we use a 100GbE fiber transceiver, our stream values are bounded by $R\approx 2^{34}$ bytes.
If we are willing to withstand an error of $\eps=2^{-20}$ (i.e., about $16\mathit{KBps}$), the work of~\cite{SWATPAPER} provides an additive approximation over the sliding window and requires about 120KB.
In contrast, using a 10 minutes slack ($\tau\triangleq\frac{1}{144}$), our algorithm for \textbf{exact} \SS{} requires only 800 bytes, 99\% less than approximate summing over exact sliding window.
For the same slack size, the algorithm of~\cite{DatarGIM02} requires more space than our \textbf{exact} algorithm even for a large 3\% error.
Further, if we also allow the same additive error ($\eps=2^{-20}$), we provide an algorithm that requires only 240 bytes - a reduction of more than $99.8\%$~!

Table~\ref{tab:results} compares our results for the important case of constant slack with~\cite{DatarGIM02}.
As depicted, our \emph{exact} algorithm is faster and more space efficient than the multiplicative approximation of~\cite{DatarGIM02}.
Comparing our multiplicative approximation algorithm to that of~\cite{DatarGIM02}, we present \emph{exponential} space reductions in the dependencies on $\oneOverE$ and $\bsrange$, with an asymptotic reduction in $\window$ as well. 
We also improve the update time from $O(\logp{RW})$ to~$O(1)$.

{\tiny
\begin{table}[t!]
	\centering
	\centering{
	\resizebox{1.0 \textwidth}{!}{
	\begin{tabular}{|c|c|c|c|c|}
		\hline
		& Exact Sum & Additive Error& \multicolumn{2}{c|}{Multiplicative Error} \\\hline
		$\tau=\Theta(1)$  & $\mathbf{\Theta(\textbf{log}\parentheses{RW})}$ & $\bm{\Theta(\textbf{log}({W/\eps}))}$ & $\bm{\Theta(\textbf{log}\parentheses{W/\eps}+\textbf{log}\textbf{log} R)}$ & $O(\oneOverE\log(RW)\log\log(RW))~\cite{DatarGIM02}$\\\hline
		Exact Window  & $\Theta(W\log{R})$ & $\Theta(\oneOverE+\log{W})~\cite{SWATPAPER}$ & $O(\oneOverE\log^2(RW))~\cite{GibbonsT02}$ &
		$O(\oneOverE\log RW\logp{W\log R})~\cite{DatarGIM02}$\\\hline
	\end{tabular}\smallskip}}
	\normalsize
	\caption{Comparison of \bs{} algorithms. Our contributions are in bold. All algorithms process elements in constant time except for the rightmost column where both update in $O(\logp{RW})$ time. We present matching lower bounds to all our~algorithms.}
	\label{tab:results}
\end{table}
}
\normalsize

Finally, we apply the slack window approach to multiple streaming problems, including \maxim{}, \gs{}, \cdp{} and \std{}.
We show that, while some of these problems cannot be approximated on an exact window in sub-linear space (e.g. maximum and general sum), we can easily do so for slack windows.
In the count distinct problem, a constant slack yields an asymptotic space reduction over~\cite{SlidingHLL,Fusy-HLL}.

\section{Preliminaries}

For $\ell\in\mathbb N$, we denote $[\ell]\triangleq\frange{\ell}$.
We consider a stream of data elements $x_1,x_2,\ldots,x_t$, where at each step a new element $x_i\in\RR$ is added to $S$.
%
 A \emph{$W$-sized window} contains only the last $W$ elements: $x_{t-W+1}\ldots x_t$. We say that $\mathcal F$ is a \emph{$\tau$-slack $W$-sized window} if there exists $c\in[\wt-1]$ such that $\mathcal F=x_{t-(W+c)+1}\ldots x_t$.
For simplicity, we assume that $\oneOverT$ and $\wt$ are integers. Unless explicitly specified, the base of all logs is~$2$.

Algorithms for the \SS{} problem are required to support two operations:
\begin{enumerate}
\item {\sc Update}$(x_t)$
	Process a new element $x_t\in\RR$.
\item
	\OUT{} $()$ Return a pair $\langle\widehat{S},c\rangle$ such that $c
	\in\mathbb N
$ is the slack size and $\widehat{S}$ is an estimation of the last $W+c$ elements sum, i.e., $S\triangleq\sum_{k=t-(W+c)+1}^t x_k$.
\end{enumerate}

We consider three types of algorithms for \SS{}:
\begin{enumerate}
\item \textbf{Exact algorithms:} an algorithm \alg{} solves \EXACT{} if its \OUT{} returns $\langle\widehat{S},c\rangle$ that satisfies $\cSet$ and~$\widehat{S}=S$.
\item \textbf{Additive algorithms:} we say that \alg{} solves \ADDI{} if its \OUT{} function returns $\langle\widehat{S},c\rangle$ that satisfies $\cSet$ and $|S -\widehat{S}| < \bserror$.
\item \textbf{Multiplicative algorithms:} \alg{} solves \MULT{} if its \OUT{} returns $\langle\widehat{S},c\rangle$ satisfying $\cSet{}$ and $\frac{S}{1+\eps} < \widehat{S}\le S$ if $S>0$, and $\widehat{S}=0$~otherwise.
\end{enumerate}


\section{Lower Bounds}
In this section, we analyze the space required for solving the \SS{} problems.
Intuitively, our bounds are derived by constructing a set of inputs that any algorithm must distinguish to meet the required guarantees.
There are two tricks that we frequently use in these lower bounds.
The first is setting the input such that the slack consists only of zeros, and thus the algorithm must return the desired approximation of the remaining window.
The next is using a ``cycle argument'' -- consider two inputs $x$ and $x\cdot y$ for $x,y\in\frange{R}^*$.
If both lead to the same memory configuration, so do such $xy^k$ for any $k\in\mathbb N$.
Thus, if there is a $k$ such that no single answer approximates $x$ and $xy^k$ well, then $x$ and $xy$ had to lead to separate memory configurations in the first place.

\subsection{\EXACT{}}\label{sec:lb-exact}
We start by proving lower bounds on the memory required for exact \SS{}.

\begin{lemma}\label{lem:RW2}
	Any deterministic algorithm \alg{} that solves the \EXACT{} problem must use at least $\ceil{\logp{RW(W+1)/2 + 1}}\ge \floor{\logp{RW^2}}$ bits.
\end{lemma}
We now use Lemma~\ref{lem:RW2}, whose proof is deferred to Appendix~\ref{app:rw2}, to show the following lower bound on \EXACT{} algorithms:

\begin{theorem}\label{thm:exactLB}
Any deterministic algorithm \alg{} that solves the \EXACT{} problem must use at least ${\max\set{\floor{\logp{RW^2}},\ceil{\ceil{\oneOverT/2}{\log \parentheses{RW\tau+1}}}}}$ bits.
\end{theorem}
\begin{proof}
Lemma~\ref{lem:RW2} shows a $\floor{\logp{RW^2}}$ bound.
We proceed with showing a lower bound\\ $\ceil{\ceil{\oneOverT/2}{\log \parentheses{RW\tau+1}}}$ bits.
Consider the following~languages:
{
	$$L_{E_2}\triangleq \set{0^{W\tau+i}\sigma R^{W\tau-i-1}\mid i\in[W\tau-1], \sigma\in[R]}, \overline{L_{E_2}}
\triangleq \set{w_1\cdot w_2\cdots w_{\ceil{\oneOverT/2}}\mid \forall i: w_i\in L_{E_2}}.$$}\normalfont
Notice that $|\overline{L_{E_2}}|=(RW\tau+1)^{\ceil{\oneOverT/2} }$ since each of the words in $L_{E_2}$ has a distinct sum of literals, and each number in $\frange{R\wt}$ is the sum of a word.
We show that each input in $\overline{L_{E_2}}$ must be mapped into a distinct memory configuration.
Let $S_1\triangleq w_{1,1}\cdot w_{2,1}\cdots w_{\ceil{\oneOverT/2},1}$, $\quad S_2\triangleq  w_{1,2}\cdot w_{2,2}\cdots w_{\ceil{\oneOverT/2},2}$ be two distinct inputs in $\overline{L_{E_2}}$ such that $\forall i:w_{i,1},w_{i,2}\in L_{E_2}$.
Denote $\chi\triangleq \max\set{i\in\brackets{\ceil{\oneOverT/2}}\mid w_{i,1}\neq w_{i,2}}$ -- the last place in which $S_1$ differs from $S_2$; also, denote $w_{\chi,1}\triangleq 0^{\wt} a, w_{\chi,2}\triangleq 0^{\wt} b$.
Consider the sequences $S^*_1 = S_1 \cdot 0^{2\wt(\chi-1/2)}$ and $S^*_2 = S_2 \cdot 0^{2\wt(\chi-1/2)}$.
Notice that the last $W$ elements windows for $S^*_1,S^*_2$ are $a\cdot w_{\chi+1,1}\cdots w_{\ceil{\oneOverT/2},1}\cdot 0^{2\wt(\chi-1/2)}$ and $b\cdot w_{\chi+1,2}\cdots w_{\ceil{\oneOverT/2},2}\cdot 0^{2\wt(\chi-1/2)}$ respectively, and that the preceding $W\tau$ elements of both are all zeros.
An illustration of the setting appears in Figure~\ref{fig:exactLB}.
\begin{figure}[t]
	\centering
\ifdefined\TWELVEPAGER
	\includegraphics[width=0.7\linewidth]{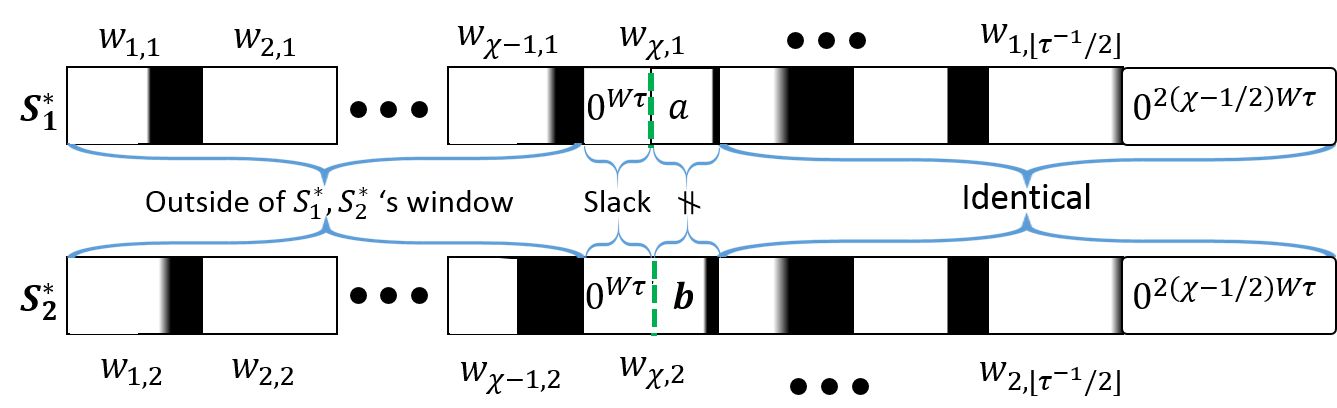}
\else	
	\includegraphics[width=\linewidth]{exact-LB.png}
\fi
	\caption{An illustration of the ${\ceil{\oneOverT/2}}{\log \parentheses{RW\tau+1}}$ lower bound setting. If we assume that after seeing $w_{1,1}\cdot w_{2,1}\cdots w_{\ceil{\oneOverT/2},1}$ we reach the same configuration as after processing $w_{1,2}\cdot w_{2,2}\cdots w_{\ceil{\oneOverT/2},2}$, then we provide a wrong answer for at least one of $S_1^*, S_2^*$.
	}
	\label{fig:exactLB}
\end{figure}
By our choice of $\chi$, we have that the sum of the last $W$ elements of $S_1^*$ and $S_2^*$ is different, and since the slack is all zeros, no answer is correct on both. Finally, note that this implies that $S_1,S_2$ had to reach different configurations, as otherwise \alg{} would reach the same configuration after processing the additional $2\wt(\chi-1/2)$~zeros.
\end{proof}

\subsection{\ADDI{}}\label{sec:lb-addi}
Next, Theorem~\ref{thm:addLB} shows a lower bound for additive approximations of \SS{}. Due to lack of space, the proof is deferred to Appendix~\ref{app:add-lb-proof}.
\begin{theorem}\label{thm:addLB}
	For $\eps < 1/4$, any deterministic algorithm \alg{} that solves the \ADDI{} problem requires $\max\set{{\log (W/\eps)}-O(1),\ceil{\ceil{\oneOverT/2}{\log \floor {{\tau/2\eps+1}}}}}$~bits.
\end{theorem}

\subsection{\MULT{}}\label{sec:lb-mult}
In this section, we show lower bounds for multiplicative approximations of \SS{}. 
We start with Lemma~\ref{lem:mult-lb1}, whose proof appears in Appendix~\ref{app:multLBLemma}.
\begin{lemma}\label{lem:mult-lb1}
	For $\eps < 1/4$, any deterministic algorithm \alg{} for the \MULT{} problem requires at least ${\log (W/\eps)+\log\logp {\bserror}}-O(1)$ memory bits.
\end{lemma}

\newcommand{\xsize}{\floor{\log RW/2}}
\renewcommand{\xsize}{\psi}
To extend our multiplicative lower bound, we use the following fact:
\begin{Fact}\label{fact1}
	For any $x\neq1,y\in\mathbb R$, the sequence {\small$\set{c_i}_{i=1}^{\infty}$}\normalfont, defined as
	$
	c_n \triangleq \begin{cases}
	1&\mbox{n = 1}\\
	x\cdot c_{n-1}+y&\mbox{Otherwise}
	\end{cases}$\\ can be represented using a closed form as $c_n=x^{n-1}+y\cdot\frac{x^n-1}{x-1}$.\\
\end{Fact}
Next, let $k\in\mathbb N$ and $\xsize,\eps\in\mathbb R$, such that $\xsize \ge 2$, $\epsilon>0$, $k\ge1$;
consider the integer sequence 
\ifdefined \TWELVEPAGER
{
	$$a_{n,k}\triangleq\begin{cases}
	1&\mbox{n = 1}\\
	\ceil{(1+\epsilon)\parentheses{a_{n-1,k}+\sum_{i=1}^{k-1}\xsize^i}}&\mbox{otherwise}.
	\end{cases}$$}
\normalfont
\else
{
$$a_{n,k}\triangleq\begin{cases}
1&\mbox{n = 1}\\
\ceil{(1+\epsilon)\parentheses{a_{n-1,k}+\sum_{i=1}^{k-1}\xsize^i}}&\mbox{Otherwise}.
\end{cases}$$}
\normalfont
\fi
Using the fact above, we show the following lemma:
\begin{lemma}\label{lem:a_nk-bound}
For every integer $n\ge 1$ we have $a_{n,k}\le 4\epsilon^{-1} {(1+\eps)^{n+1}}{\xsize^{k-1}}$.
\end{lemma}
\begin{proof}
To apply Fact~\ref{fact1}, we define an upper bounding sequence $\set{b_{i,k}}_{i=1}^\infty$ as follows:
{
$$b_{n,k}\triangleq\begin{cases}
1&\mbox{n = 1}\\
{(1+\epsilon)\parentheses{b_{n-1,k}+\sum_{i=1}^{k-1}\xsize^i}}+1&\mbox{Otherwise}.
\end{cases}$$}
\normalfont
Thus, we can rewrite the $n$'th element of the sequence as:
$$\resizeIfTWELVEPAGER{b_{n,k} = (1+\epsilon)^{n-1}+\frac{(1+\eps)^n-1}{(1+\epsilon)-1}\parentheses{(1+\epsilon)\sum_{i=1}^{k-1}\xsize^i + 1}.}
$$
We can now use this representation to derive an upper bound of $b_{n,k}$:
{
\begin{align*}
b_{n,k} &
\resizeIfTWELVEPAGER{= (1+\epsilon)^{n-1}+\parentheses{(1+\epsilon)\sum_{i=1}^{k-1}\xsize^i + 1}\frac{(1+\eps)^n-1}{(1+\epsilon)-1} }
\\
&\le (1+\epsilon)^{n-1}+\parentheses{(1+\epsilon)2\xsize^{k-1}}\frac{(1+\eps)^n-1}{\epsilon}  \le 4\epsilon^{-1} {(1+\eps)^{n+1}}{\xsize^{k-1}}.
\end{align*}}%
Finally, since $a_{n,k}\le b_{n,k}$ for any $n,k$, we conclude that $a_{n,k} \le 4\epsilon^{-1} {(1+\eps)^{n+1}}{\xsize^{k-1}}$.
\end{proof}
We now define the integer set $I_k$ as $I_k \triangleq \set{a_{n,k}\mid a_{n,k} \le \xsize^k}$, and proceed to bound $|I_k|$.
\begin{lemma}\label{lem:I_nk-bound}
	For any $k\ge 1$ we have $|I_k|\ge \epsilon^{-1}\lnp{{\xsize\eps/4}}-1$.
\end{lemma}
\begin{proof}
Clearly, the cardinality of $I_k$ is the largest $n$ for which $a_{n,k} \le \xsize^k$. According to Lemma~\ref{lem:a_nk-bound}, we have that $a_{n,k}\le 4\epsilon^{-1} {(1+\eps)^{n+1}}{\xsize^{k-1}}$, and thus:
\begin{multline*}
|I_k|=\arg\max\set{n\mid 4\epsilon^{-1} {(1+\eps)^{n+1}}{\xsize^{k-1}}\le\xsize^k}\\\ge \log_{1+\eps}\parentheses{{\xsize\eps/4}}-1
= \frac{\lnp{{{\xsize\eps/4}}}}{\lnp{1+\eps}}-1\ge \epsilon^{-1}\lnp{{\xsize\eps/4}}-1.\qquad\qedhere
\end{multline*}
\end{proof}
We proceed with a stronger lower bound for non-constant $\tau$ values.
\newcommand{\nwords}{\ceil{\oneOverT/2}}
\newcommand{\mintau}{\sqrt{\frac{4}{RW}}}
\begin{lemma}\label{lem:multLB2}
	For $\frac{1}{2\logp{RW}-8}\le\tau\le 1$, any deterministic algorithm \alg{} that solves \MULT{} requires at least $\Omega\parentheses{\oneOverT\parentheses{\logp{\tau/\eps} + \log\logp{RW}}}$~bits.
\end{lemma}
\begin{proof}
We use  $rep(x)\triangleq(x\mod R)\cdot R^{\floor{x/R}}$ to denote a sequence in $\set{\sigma R^*\mid \sigma\in[R]}$ that has a sum of $x$.
For an integer set $I_k$, we denote $rep(I_k)\triangleq\set{rep(x)\mid x\in I_k}$.
We now choose the value of $\xsize$ to be $\xsize\triangleq \sqrt[\leftroot{-1}\uproot{3}\nwords]{RW/8}$; notice that $\xsize\ge 2$ as required.
Next, consider:
{
\begin{multline*}
\overline{L_{M,2}}\triangleq 0^W \cdot 0^{\wt}\cdot rep(I_{\nwords}) \cdot 0^{\wt}\cdot rep(I_{\nwords-1}) \cdots 0^{\wt}\cdot rep(I_1)\\
= \set{0^W\cdot w_1\cdot w_2\cdots w_{\nwords}\mid \forall i:
	w_i\in \set{0^{\wt}\cdot rep(x)\mid x\in  I_{\nwords+1-i}}}.
\end{multline*}}
\normalfont

That is, every word in the $\overline{L_{M,2}}$ language consists of a concatenation of words $w_1,\ldots,w_{\nwords}$, such that every $w_i$ starts with $\wt$ zeros followed by a string representing an integer in $I_{\nwords+1-i}$, which is defined above.
According to Lemma~\ref{lem:I_nk-bound} we have that
{
\begin{align*}
\log(|&\overline{L_{M,2}}|) \ge \logp{\parentheses{{\epsilon^{-1}\lnp{{\xsize\eps/4}}-1}}^{\nwords}}=\nwords\parentheses{\log\oneOverE+\log\logp{\xsize\eps}-O(1)}\\
& = \Omega\parentheses{\oneOverT\parentheses{\log\oneOverE + \log\logp{{\sqrt[\leftroot{-1}\uproot{3}\nwords]{RW/8}\cdot\eps}}}}\\
& = \Omega\parentheses{\oneOverT\parentheses{\log\oneOverE + \logp{\frac{\logp{{{RW/8}}}}{\nwords}+\log\eps}}}\\
&
= \Omega\parentheses{\oneOverT\parentheses{\logp{\tau/\eps} + \log\logp{RW}}}.
\end{align*}}%
\normalfont
Next, we show that every two words in $\overline{L_{M,2}}$ must reach different memory configurations, thereby implying a $\Omegap{\logp{|\overline{L_{M,2}}|}}$ bits lower bound.
Let $S_1\neq S_2\in\overline{L_{M,2}}$ such that $S_1 = 0^W\cdot w_{1,1}\cdots w_{\nwords,1}$, $S_2 = 0^W\cdot w_{1,2}\cdots w_{\nwords,2}$, and
$\forall i\in\set{1,\ldots,\nwords}j\in\set{1,2}:w_{i,j}\in \set{0^{\wt}\cdot rep(x)\mid x\in  I_{\nwords+1-i}}$. We next assume by contradiction that $S_1$ and $S_2$ leads \alg{} to the same memory configuration.
Let {$\chi\in\set{1,\ldots,\nwords}$} such that $w_{\chi,1}\neq w_{\chi,2}$. Since \alg{} reaches an identical configuration after reading $S_1,S_2$, and as it is deterministic, \alg{} must reach the same configuration when processing $S_1\cdot 0^{2\wt(\chi-1/2)}$ and $S_2\cdot 0^{2\wt(\chi-1/2)}$.
Next, observe that for every $k\in\{1,\ldots,\nwords\}$, the representation length of any of its words is bounded by $\ceil{\xsize^k / R}$.
Thus, the length of a word in
	$$\set {w_1\cdot w_2\cdots w_{\nwords}\mid \forall i:
	w_i\in \set{0^{\wt}\cdot rep(x)\mid x\in  I_{\nwords+1-i}}} \text{ is at most}$$
\normalfont
\begin{multline*}
\sum_{k=1}^{\nwords}\ceil{\wt + \xsize^k / R}
\le \nwords(\wt+1) + {2\xsize^{\nwords} / R}\\
= \nwords(\wt+1) + 2{W/8} \le 3W/4 + \nwords + \wt\le W+\wt.
\end{multline*}%
\normalfont
Now, since every word $w_{i,j}$ starts with a sequence of $\wt$ zeros, the slack size chosen by the algorithm is irrelevant and the sums the algorithm must estimate are $\sum_{i=\chi}^{\nwords}s(w_{i,1})$ and $\sum_{i=\chi}^{\nwords}s(w_{i,2})$, where $s(w_{i,j})$ is simply the sum of the symbols in $w_{i,j}$.
Note that $s(w_{\chi,1})$ and $s(w_{\chi,2})$ are integers in $I_{\nwords+1-\chi}$.
We assume without loss of generality that $s(w_{\chi,1})< s(w_{\chi,2})$  (i.e., $s(w_{\chi,1})< s(w_{\chi,2})\in I_{\nwords+1-\chi}$).
Finally, it follows~that
$$
{\sum_{i=\chi}^{\nwords}s(w_{i,1}) \le s(w_{\chi,1}) + \sum_{i=\chi+1}^{\nwords}\max(I_{\nwords+1-i}) \le s(w_{\chi,1}) + \sum_{k=1}^{\chi-1}\xsize^k\le \frac{s(w_{\chi,2})}{1+\eps},}
$$
where the last inequality follows from the definition of $I_{\nwords+1-\chi}$.
Thus, no $\widehat{S}$ value is correct for both $S_1\cdot 0^{2\wt(\chi-1/2)}$ and $S_2\cdot 0^{2\wt(\chi-1/2)}$.
\end{proof}
Finally, we combine Lemma~\ref{lem:mult-lb1} and Lemma~\ref{lem:multLB2} to obtain the following lower bound:
\begin{theorem}\label{thm:multLB}
	For $\eps < 1/4, \frac{1}{2\logp{RW}-8}\le\tau\le 1$, any deterministic algorithm for the \MULT{} problem requires at least $\Omega\big(\log (W/\eps)\allowbreak+ \oneOverT\parentheses{\logp{\tau/\eps} + \log\logp{RW}}\big)$~bits.
\end{theorem}

\section{Upper Bounds}
In this section, we introduce solutions for the \SS{} problems.
In general, all our algorithms have a structure that consists of a subset of the following, where ``compression'' has a different meaning for the exact, additive and multiplicative variants:
\begin{itemize}
\item Compress the arriving item.
\item Add the item into a counter $y$ and compress the counter.
\item If a $\wt$-sized \emph{block} ends, store it as a compressed representation of $y$. Sometimes we propagate the compression error to the following block; otherwise, we zero $y$.
\item Use the block values and $y$ to construct an estimation for the sum.
\end{itemize}
Our \emph{double rounding} technique, described below, asymptotically improves over running $1/\tau$ separate plain stream (insertion only) algorithm instances.


\subsection{\EXACT{}}\label{sec:exact-slacky}
We divide the stream into ${\wt}$-sized blocks and sum the number of arriving elements in each block with a $\ceil{\logp{R\wt+1}}$ bits counter.
We maintain the sum of the current block in a variable called $\remainder$, $\blockOffset$ maintains the number of elements within the current block, and $\currentBlockIndex$ is the current block number. The variable $\bitarray$ is a cyclic buffer of $\oneOverT$ blocks. Every $W\tau$ steps, we assign the value of $\remainder$ to the oldest block ($\currentBlock$) and increment $\currentBlockIndex$.
Intuitively, we ``forget'' $\currentBlock$ when its block is no longer part of the window.
To satisfy queries in constant time, we also maintain the sum of all active counters in a $\ceil{\logp{RW(1+\tau)+1}}$-bits variable named $\sumOfBits$.
Algorithm~\ref{alg:exact} provides pseudocode for the described algorithm.
\begin{algorithm}[]
	\algsize{}
	\caption{\EXACT{} Algorithm}\label{alg:exact}
	\begin{algorithmic}[1]
		\Statex Initialization: $\remainder = 0, \bitarray = \bar0, \sumOfBits = 0, \currentBlockIndex=0, \blockOffset=0$.
		\Function{\add[\inputVariable]}{}
		\State $\remainder \gets \remainder + \inputVariable$		
		\State $\blockOffset\gets (\blockOffset+ 1) \mod \wt$		
		\If {$\blockOffset = 0$} \Comment{End of block}
		\State $\sumOfBits \gets \sumOfBits - \currentBlock + \remainder$
		\State $b_i \gets \remainder$
		\State $\remainder \gets 0$		
		\State $\currentBlockIndex\gets (\currentBlockIndex+1) \mod \oneOverT$	
		\EndIf
		\EndFunction
		
		\Function{\query}{}
		\State \Return {$\langle\sumOfBits+\remainder, c\rangle$}
		\EndFunction
		
	\end{algorithmic}
\end{algorithm}
\normalsize
\noindent We now analyze the memory consumption of Algorithm~\ref{alg:exact}.
\begin{theorem}\label{thm:exactMem}
Algorithm~\ref{alg:exact} uses $(\oneOverT+1)\ceil{\logp{R\wt+1}}+\logp{RW^2} + O(1)$ bits.
\end{theorem}
\begin{proof}
$\remainder$ takes $\clogp{R\wt+1}$ bits; $\sumOfBits$ requires $\clogp{RW+1}$; $\currentBlockIndex$ adds $\clog\oneOverT$ bits, while $\blockOffset$ needs $\clog{\wt}$ bits. Finally, $\bitarray$ is a $\oneOverT$-sized array of counters, each allocated with $\clogp{R\wt+1}$ bits. Overall,
\ifdefined \TWELVEPAGER
it uses
\else
Algorithm~\ref{alg:exact} uses at most
\fi
$(\oneOverT+1)\ceil{\logp{R\wt+1}}+\logp{RW^2}+4$~bits.
\end{proof}
We conclude that Algorithm~\ref{alg:exact} is asymptotically optimal.
\begin{theorem}
	\label{thm:4Exact}
	Let $\mathcal B\triangleq{\max\set{\floor{\logp{RW^2}},\ceil{\ceil{\oneOverT/2}{\log \parentheses{RW\tau+1}}}}}$ be the \EXACT{} lower bound of Theorem~\ref{thm:exactLB}. Algorithm~\ref{alg:exact} uses at most $\mathcal B (4+o(1))$ memory~bits.
\end{theorem}
\ifdefined \TWELVEPAGER
\else
\begin{proof}	
As shown in Theorem~\ref{thm:exactMem}, the number of bits used by Algorithm~\ref{alg:exact} is
{\footnotesize
\begin{multline*}
(\oneOverT+1)\ceil{\logp{R\wt+1}}+\logp{RW^2} + O(1)\\
\le \oneOverT\ceil{\logp{R\wt+1}}+2\logp{RW^2}+O(1) \le \mathcal B (4+o(1))
\qquad\qedhere
\end{multline*}}%
\normalfont
\end{proof}
\fi
Theorem~\ref{thm:4Exact} shows that Algorithm~\ref{alg:exact} is only x$4$ larger than the lower bound.
In Appendix~\ref{app:tighter} we show that in some cases we can get considerably closer to the lower bound.
Finally, in Appendix~\ref{app:exact-alg-prof} we show that Algorithm~\ref{alg:exact} is correct.


\subsection{\ADDI{}}\label{sec:addi-summing}
We now show that additional memory savings can be obtained by combining slackness with an additive error.
First, we consider the case where $\tau\le 2\epsilon$.
In~\cite{SWATPAPER}, we proposed an algorithm that sums over (exact) $W$ elements window using the optimal $\Theta(\oneOverE+\logw)$ bits, with an additive error of $\bserror$.
Next, notice that if an algorithm solves \ADDI{}, it also solves {$(W,\tau,\tau/2)${\sc -Additive Summing}}; hence, we can apply Theorem~\ref{thm:addLB} to conclude that it requires $\Omega(\oneOverT+\logw)=\Omega(\oneOverE+\logw)$. Thus, we can run the algorithm from~\cite{SWATPAPER} and remain asymptotically memory optimal with no slack at all!

Henceforth, we assume that
$\tau > 2\eps$;
we present an algorithm for the problem using a \emph{$2$-stage rounding} technique.
When a new item arrives, we scale it by $R$ and then round the results to $O(\log\oneOverE)$ bits.
As in Section~\ref{sec:exact-slacky}, we break the stream into non-overlapping blocks of size $\wt$ and compute the sum of each block separately.
However, we now sum the rounded values rather than the exact input, with a $O(\log\frac{\wt}{\eps})$-bits counter denoted $y$.
Once the block is completed, we \emph{round its sum} such that it is represented with $O(\log\frac{\tau}{\eps})$ bits.
Note that this second rounding is done for the entire block's sum while we still have the ``exact'' sum of rounded fractions.
Thus, we \emph{propagate} the second rounding error to the following block.
An illustration of our algorithm appears in Figure~\ref{fig:slackySumming}.
Here, $\text{Round}_{\bsReminderFractionBits}(z)$ refers to rounding a fractional number $z\in[0,1]$ into the closest number $\widetilde{z}$ such that $2^\bsReminderFractionBits\cdot\widetilde{z} \in \mathbb N$.
Algorithm~\ref{alg:addi} provides pseudo code for the algorithm, which uses the following variables:
\begin{enumerate}
\item $\remainder$ - a fixed point variable that uses $\ceil{\log{\wt}}+1$ bits to store its integral part and additional $\bsReminderFractionBits_1 \triangleq \ceil{\log\oneOverE} + 1$ bits for storing the fractional part.
\item $\bitarray$  - a cyclic array that contains $\oneOverT$ elements, each of which takes $\bsReminderFractionBits_2\triangleq\ceil{\log\frac{\tau}{\eps}}$ bits.
\item $\sumOfBits$ - keeps the sum of elements in $\bitarray$ and is represented using $\logp{\oneOverT\ceil{\log\frac{\tau}{\eps}}+1}$ bits.
\item $\currentBlockIndex$ - the index variable used for tracking the oldest block in $\bitarray$.
\item $\blockOffset$ - a variable that keeps the offset within the ${\wt}$ sized block.
\end{enumerate}


\begin{figure}[]
	\centering
\ifdefined\TWELVEPAGER	
	\includegraphics[width=0.7\linewidth]{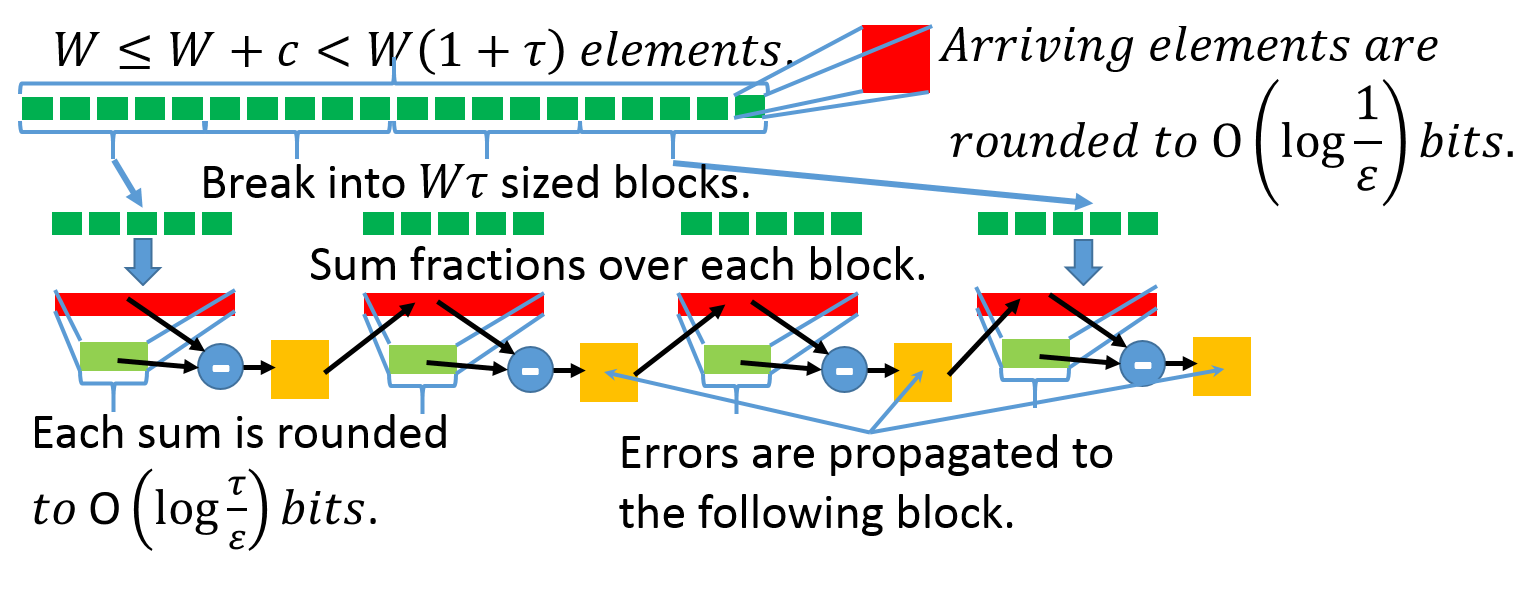}
\else
	\includegraphics[width=0.8\linewidth, height=4cm]{slackySumming.png}
\fi
	\caption{An illustration of our 2-stage rounding technique. Arriving elements are rounded to $\parentheses{\clog\oneOverE+1}$ bits. We then sum the rounded fractions of each block and round the resulting sum into $\clog{\frac{\tau}{\eps}}$ bits. The second rounding error is propagated to the next~block.
	}
	\label{fig:slackySumming}
\vspace*{-0.2cm}
\end{figure}
\begin{algorithm}[t]
	\caption{\ADDI{} Algorithm}\label{alg:addi}
	\algsize
	\begin{algorithmic}[1]
		\Statex Initialization: $\remainder = 0, \bitarray = 0, \sumOfBits = 0, \currentBlockIndex=0, \blockOffset=0$.
		\Function{\add[\inputVariable]}{}
		\State $\bsFracInput \gets \text{Round}_{\bsReminderFractionBits_1}\parentheses{\frac{\inputVariable}{R}}$ \label{line:rounding} \Comment{Round $\parentheses{\frac{\inputVariable}{R}}$ such that $x'\cdot 2^{\bsReminderFractionBits_1}\in\mathbb N$}
		\State $\remainder \gets \remainder + \bsFracInput$		
		\State $\blockOffset\gets (\blockOffset+ 1) \mod \wt$		
		\If {$\blockOffset = 0$}\label{line:end-of-block} \Comment{End of block}
			\State $\sumOfBits \gets \sumOfBits - \currentBlock$
			\State $b_i \gets \text{Round}_{\bsReminderFractionBits_2}(\frac{y}{{\wt}})$ \Comment{Replace the value for the block that has left the window.}
			\State $B \gets B + b_i$			
			\State $y\gets y - {\wt}\cdot b_i$
			\State $i\gets (i+1) \mod \numBlocks$	
		\EndIf
		\EndFunction
		
		\Function{\query}{}
		\State \Return {$\langle\bsrange\cdotpa {\wt \cdot B + y}, c\rangle$}		
		\EndFunction
		
	\end{algorithmic}
\end{algorithm}
\normalsize
\noindent We now analyze the memory consumption of Algorithm~\ref{alg:addi}.
\begin{theorem}\label{thm:addiMem}
	Algorithm~\ref{alg:addi} uses $\oneOverT\logp{\frac{\tau}{\eps}}(1+o(1))+2\log(W/\eps)$ bits.
\end{theorem}
\begin{proof}
$\remainder$ requires $\logp{\frac{\wt}{\eps}}+O(1)$ bits; $b$ requires another $\oneOverT\clogp{\frac{\tau}{\eps}}$; $B$ takes additional $\logp{\oneOverT\ceil{\log\frac{\tau}{\eps}}+1}$ bits; $i$ adds $\ceil{\log\oneOverT}$ bits, while and $\blockOffset$ is represented with $\clog\wt$ bits. Overall, the space requirement is $\oneOverT\clogp{\frac{\tau}{\eps}}(1+o(1))+2\log(W/\eps)$ bits.
\end{proof}
\begin{corollary}
	Let $\mathcal B\triangleq\max\set{{\log (W/\eps)}-O(1),\ceil{\ceil{\oneOverT/2}{\log \floor {{\tau/2\eps+1}}}}}$ be the \ADDI{} space lower bound of Theorem~\ref{thm:addLB}, then Algorithm~\ref{alg:addi} uses $\mathcal B\cdotpa{4+o(1)}$~bits.
\end{corollary}
Finally, Theorem~\ref{thm:algadi} shows that Algorithm~\ref{alg:addi} is correct. The proof is deferred to Appendix~\ref{app:AddAlgProof}
\begin{theorem}

Algorithm~\ref{alg:addi}  solves the \ADDI{} problem.
	\label{thm:algadi}
\end{theorem}
\newcommand{\logBase}{\ensuremath{(1+\epsilon/2)}}
\newcommand{\roundDown}[1]{\parentheses{#1}_\downarrow}
\subsection{\MULT}\label{sec:mult}
In this section, we present Algorithm~\ref{alg:mult} that provides a $(1+\epsilon)$ multiplicative approximation of the \SS{} problem.
Compared to Algorithm~\ref{alg:exact}, we achieve a space reduction by representing each sum of $\wt$ elements using $O(\log\logp{ RW\tau} + \log{\oneOverE})$ bits. Specifically, when a block ends, if its sum was $y$, we store
$\rho = \floor{\log_{\logBase}y}$ (we allow a value of $-\infty$ for $\rho$ if $y=0$). To achieve $O(1)$ \query{}, we also store an approximate window sum $B$, which is now a \emph{fixed point} fractional variable with $O(\log RW)$ bits for its integral part and additional $O(\log \oneOverE)$ bits for storing a fraction. To update $B$'s value for a new $\rho$, we \emph{round down} the value of ${(1+\epsilon)^\rho}$.
Specifically, for a real number $x$, we denote $\roundDown x\triangleq \floor{x\cdot k}/k$, for $k\triangleq \ceil{4\over \epsilon}$.
Our pseudo code appears in Algorithm~\ref{alg:mult}.  The algorithm requires $O\big(\oneOverT\parentheses{\log\logp{ RW\tau} + \log{\oneOverE}}\allowbreak+\log RW\big)$ bits of space and is memory optimal when $R=W^{O(1)}$ and $\tau=\Omegap{\frac{1}{\log RW}}$.
The full analysis of Algorithm~\ref{alg:mult} is deferred to Appendix~\ref{app:mult-lb-proof}.
\begin{algorithm}[H]
	\algsize{}
	\caption{\MULT{} Algorithm}\label{alg:mult}
	\begin{algorithmic}[1]
		\Statex Initialization: $\remainder = 0, \bitarray = \bar0, \sumOfBits = 0, \currentBlockIndex=0, \blockOffset=0$.
		\Function{\add[\inputVariable]}{}
		\State $\remainder \gets \remainder + \inputVariable$		\label{line:mult-exact-inblock-summing}
		\State $\blockOffset\gets (\blockOffset+ 1) \mod \wt$		
		\If {$\blockOffset = 0$} \Comment{End of block}
		\State $\rho \gets \floor{\log_{\logBase}y}$ \Comment{If $y=0$ we use $\rho=-\infty$ and $\logBase^\rho=0$} \label{line:mult-rho}
		\State $\sumOfBits \gets \sumOfBits - \roundDown{{\logBase^{\currentBlock}}} + \roundDown{{\logBase^\rho}}$  \label{line:mult-B-summing}
		\State $b_i \gets \rho$
		\State $\remainder \gets 0$		\label{line:mult-y-reset}
		\State $\currentBlockIndex\gets (\currentBlockIndex+1) \mod \oneOverT$	
		\EndIf
		\EndFunction
		
		\Function{\query}{}
		\State \Return {$\langle\sumOfBits+\remainder, c\rangle$}
		\EndFunction
		
	\end{algorithmic}
\end{algorithm}
\normalsize
Next, we present an alternative \MULT{} algorithm that achieves optimal space consumption for $\tau=\Theta(1)$, regardless of the value of $R$.
\paragraph*{Improved \MULT{} for $\tau=\Theta(1)$\\}\label{apx:mult2}
Algorithm~\ref{alg:mult2} is more space efficient than  Algorithm~\ref{alg:mult} but has a query time of $O(\oneOverT)$.
For $\tau=\Theta(1)$,
Algorithm~\ref{alg:mult2} is memory optimal \emph{and} supports constant time queries even if $R=W^{\omega(1)}$; for this case, Algorithm~\ref{alg:mult} requires $\Omega(\log R)$ bits which is sub~optimal.

Intuitively, we shave the $\Omegap{\log R}$ bits from the space requirement of Algorithm~\ref{alg:mult} using an approximate representation for our $\remainder$ variable and by not keeping the $\sumOfBits$ variable that allowed $O(1)$ time queries regardless of the value of $\tau$.
To avoid using $\Omegap{\log R}$ bits in $\remainder$, we use a \emph{fixed point} representation in which $O(\log\oneOverE+\log\logp{R\wt})$ bits are allocated for its integral part and another $O(\log\wt)$ for the fractional part.
The goal of $\remainder$ is still to approximate the sum of the elements within a block, but now we aim for the sum to be approximately $\altLogBase^{\remainder}$.
Whenever a block ends, we store only the integral part of $\remainder$ in our cyclic array $\bitarray$ to save space.
When queried, we compute an estimate for the sum using all of the values in $\bitarray$, which makes our query procedure take $O(\log\oneOverT)$ time.
To use the fixed point structure of $y$, we use the operator $\remRoundDown{\cdot}$ that rounds a real number $x$ into $\remRoundDown x\triangleq \floor{x\cdot \wt}/\wt$.
We denote $\log_{\altLogBase}\parentheses{0}=-\infty, \remRoundDown{-\infty}=-\infty, \floor{-\infty}=-\infty$ and $\altLogBase^{-\infty}=0$.
In appendix~\ref{apx:mult2} we prove the following theorem.
\begin{algorithm}[]
	\algsize{}
	\caption{\MULT{} Algorithm for $\tau=\Theta(1)$}\label{alg:mult2}
	\begin{algorithmic}[1]
		\Statex Initialization: $\remainder = -\infty, \bitarray = \bar0, \currentBlockIndex=0, \blockOffset=0$.
		\Function{\add[\inputVariable]}{}
		\State $\remainder \gets \remRoundDown{\log_{\altLogBase}\parentheses{{x+\altLogBase^{\remainder}}}}$
		\label{line:mult-exact-inblock-summing2}
		\State $\blockOffset\gets (\blockOffset+ 1) \mod \wt$		
		\If {$\blockOffset = 0$} \Comment{End of block}
		\State $b_i \gets \floor\remainder$\label{line:floorY2}
		\State $\remainder \gets -\infty$		\label{line:mult-y-reset2}
		\State $\currentBlockIndex\gets (\currentBlockIndex+1) \mod \oneOverT$	
		\EndIf
		\EndFunction
		
		\Function{\query}{}
		\State \Return {$\left\langle\altLogBase^{\remainder} + \sum_{i=0}^{{\oneOverT}-1} \altLogBase^{b_i}, c\right\rangle$}
		\EndFunction
		
	\end{algorithmic}
\end{algorithm}
\normalsize
\begin{theorem}
	For $\tau=\Theta(1)$, Algorithm~\ref{alg:mult2} processes elements and answers queries in $O(1)$ time, uses $O(\log(W/\eps)+\log\log R)$ bits, and is asymptotically optimal.
\end{theorem} 

\subsection{The Mean of a Slack Window}
For some applications there is value in knowing the \emph{mean} of a slack window. For example, a load balancer may be interested in the average transmission throughput. In exact windows, the sum and the mean can be derived from each other as the window size is constant. In slack windows, the window size changes but our algorithms also return the current slack offset $\cSet$. That is, by dividing $\widehat{S}$ by $W+c$ we get an estimation of the mean (we assume that stream size is larger than W).
Specifically, Algorithm~\ref{alg:exact} provides the exact mean; Algorithm~\ref{alg:addi} approximates it with $R\epsilon$ additive error,  while Algorithm~\ref{alg:mult} yields a $(1+\varepsilon)$ multiplicative approximation. 
\section{Other Measurements over Slack Windows}
We now explore the benefits of the slack model for other problems.

\textbf{\maxim.\quad{}}
While maintaining the maximum of a sliding window can be useful for applications such as anomaly detection~\cite{IntrusionDetection,IntrusionDetection2}, tracking it over an exact window is often infeasible. Specifically, any algorithms for a maximum over an (exact) window must use $\Omegap{W\logp{ R/W}}$ bits~\cite{DatarGIM02}.
The following theorem, proved in Appendix~\ref{app:max-thm} shows that
we can get a much more efficient algorithm for slack windows.
Observe the the following bounds match for $\tau$ values that are not too small ($\tau=R^{\Omegap{1}-1}$).
\begin{theorem}\label{thm:max-thm}
Tracking the maximum over a slack window deterministically requires $O\parentheses{\oneOverT\log R}$ and $\Omegap{\oneOverT\log R\tau}$ bits.
\end{theorem}

\textbf{\std.\quad{}}
Building on the ability of our summing algorithms to provide the size of the slack window that they approximate, we can compute standard deviations over slack windows. Intuitively, the standard deviation of the window can be expressed as
{\scriptsize
$$\sigma_{\overline{W}}\triangleq\sqrt{\frac{\sum_{x\in {\overline{W}}}(x-m_{\overline{W}})^2}{\left|\overline{W}\right|-1}}
=\sqrt{\frac{\sum_{x\in {\overline{W}}}x^2-2m_{\overline{W}}\sum_{x\in {\overline{W}}}x+\overline{W}\cdot m_{\overline{W}}^2}{\left|\overline{W}\right|-1}}
=\sqrt{\frac{\sum_{x\in \overline{W}}x^2-\overline{W}\cdot m_{\overline{W}}^2}{\left|\overline{W}\right|-1}},$$
}\normalfont
there $\overline{W}$ is the slack window and $m_{\overline{W}}$ is its mean. We can then use two slack summing instances to track $\sum_{x\in \overline{W}}x^2$ and $m_{\overline{W}} = |\overline{W}|^{-1}\sum_{x\in \overline{W}}x$. This gives us an algorithm that computes the exact standard deviation over slack windows using $O(\tau^{-1}\logp{RW\tau})$ space. Similarly, by using approximate rather than exact summing solutions we can compute a $(1+\eps)$ multiplicative approximation for the standard deviation using $O\big(\oneOverT\big(\log\oneOverE+\log\logp{R\wt}\big)+\logw\big)$ bits, or an $R\eps$-additive approximation using $O(\oneOverT\logp{\frac{\tau}{\eps}}+\logw)$ space. We expand on this further in Appendix~\ref{app:standardDev}.

\textbf{\gs.\quad{}}
\gs{} is similar to \bs{}, except that the integers can be in the range $\set{-R,\ldots,R}$. That is, we now allow for negative elements as well. Datar et al.~\cite{DatarGIM02} proved that General Sum requires $\Omega(W)$ bits, even for $R=1$ and constant factor approximation. In contrast, our exact summing algorithm from section~\ref{sec:exact-slacky} trivially generalizes to \gs{} and allows exact solution over slack windows.

\textbf{\cdp.\quad{}}
Estimating the number of \textbf{distinct} elements in a stream is another useful metric.
In networking, the packet header is used to identify different flows, and it is useful to know how many distinct of them are currently active.
A sudden spike in the number of active flows is often an indication of a threat to the network.
It may indicate the propagation of a worm or virus, port scans that are used to detect vulnerabilities in the system and even \emph{Distributed Denial of Service (DDoS)} attacks~\cite{Chandola07anomalydetection,CD0,Ganguly2007}.

Here, we have studied the memory reduction that can be obtained by following a similar flow to our summing algorithms -- we break the stream into $\wt$ sized blocks and run the state of the art approximation algorithm on each block separately. Luckily, count distinct algorithms are \emph{mergable}~\cite{mergable}. That is, we can merge the summaries for each block to obtain an estimation of the number of distinct items in the union of the blocks. In Appendix~\ref{app:countDistinct} we show that this approach yields an algorithm with superior space and query time compared to the state of the art algorithms for counting distinct elements over sliding windows~\cite{SlidingHLL,Fusy-HLL}. Formally, we prove the following theorem.

\begin{theorem}
For $\tau=\Theta(1)$ and any fixed $m>0$, there exists an algorithm that uses $O(m)$ space, performs updates in constant time and answers queries in time $O(m)$, such that the result approximates a window whose size is in $[W,W(1+\tau)]$; the resulting estimation is asymptotically unbiased and has a standard deviation of $\sigma = O(\frac{1}{\sqrt{m}})$.
State of the art approaches for exact windows~\cite{SlidingHLL,Fusy-HLL} require $O(m\logp{W/m})$ space and $O(m\logp{W/m})$ time per query for a similar standard deviation.
\end{theorem}

\section{Discussion}
In this work we have explored the slack window model for multiple streaming problems.
We have shown that it enables asymptotic space and time improvements.
Particularly, introducing slack enables logarithmic space exact algorithms for certain problems such as \maxim{} and \gs.
In contract, these problems do not admit sub-linear space \emph{approximations} in the exact window model.
Even in problems that do have sub-linear space approximations such as \std{} and \cdp, adding slack asymptotically improves the space requirement and allows for constant time updates.

Much of our work has focused on the classic \bs{} problem.
Based on our findings, we argue that allowing a slack in the window size is an attractive approximation axis as it enables greater space reductions compared to an error in the sum.
As an example, for a fixed $\eps$ value, computing a $(1+\eps)$-multiplicative approximation requires $\Omega(\logp{RW}\log{W})$ space~\cite{DatarGIM02}. Conversely, a $(1+\tau)$ multiplicative error in the window size, for a constant $\tau$, allows summing using $\Theta(\logp{RW})$ bits -- same as in summing $W$ elements without sliding windows!
Given that for exact windows randomized algorithms have the same asymptotic complexity as deterministic ones~\cite{SWATPAPER,DatarGIM02}, we expect randomization to have limited benefits for slack windows as well.



\newpage
{
	\bibliographystyle{plain}
	\bibliography{references}
}

\newpage
\appendix

\section{Proof of Lemma~\ref{lem:RW2}}\label{app:rw2}
\begin{proof}
Consider the following language
$$L_{E_1}\triangleq \set{0^{W\tau+i}\sigma R^{W-i-1} 0^j\mid i,j\in[W-1], i\ge j,\sigma\in([R]\setminus\set{0})} \cup \{0^{W+\wt}\}.$$
That is, $L_{E_1}$ contains a word with $W + W\tau$ consecutive zeros and the rest of the words in $L_{E_1}$ are composed of these components in this order:
\begin{itemize}
	\item $W\tau +i$ zeros for some $i\in[W-1]$.
	\item a non zero symbol $\sigma$.
	\item $W-i-1$ repetitions of the maximal symbol ($R$).
	\item $j$ zeros for some $j\in [i]$.
\end{itemize}

Our lower bound stems from the observation that every word in $L_{E_1}$ must lead to a different state. The language size is:
$|L_{E_1}| = 1+\sum_{i=0}^{W-1}R(i+1)=1+RW(W+1)/2.$
Therefore, the number of required bits is at least: $\ceil{\log|L_{E_1}|}>\left(\log(RW^2)-1\right)$.
Further, this number is an integer and therefore at least $\floor{\log(RW^2)}$ bits are required.

First, notice that the word composed of $W + W\tau$ zeros requires a unique configuration as $\mathbb{A}$ must return $0$ after processing that word.
In contrast, it must not return $0$ after processing any other word as there is at least a single $R$ within the last $W$ elements.

Let $w_1, w_2 \in L_{E_1}$ be two different words that are not all-zeros.
We need to show that $w_1$ and $w_2$ require different memory configuration.

By definition of $L_{E_1}$,  $w_1 = 0^{W\tau+i_1}\sigma_1 R^{W-i_1-1} 0^{j_1}$ and $w_2 = 0^{W\tau+i_2}\sigma_2 R^{W-i_2-1} 0^{j_2}$.
Observe that the last $W$ elements of $w_1,w_2$ are $ 0^{i_1-j_1}\sigma_1 R^{W-i_1-1} 0^{j_1}$ and $ 0^{i_2-j_2}\sigma_1 R^{W-i_2-1} 0^{j_2}$ respectively and that both are preceded with at least $\wt$ zeros.
If $i_1 \neq i_2$ or $\sigma_1\neq\sigma_2$, then $\sigma_1+R\cdotpa{W-i_1-1} \neq \sigma_2+R\cdotpa{W-i_2-1}$ and thus $\mathbb{A}$ cannot return the same count for both,
regardless of the slack, as it is all zeros ib both $w_1$ and $w_2$.

Next, assume that $i_1=i_2$ , $\sigma_1 = \sigma_2$ and that without loss of generality $j_1 < j_2$.
This means that both $w_1$ and $w_2$ have the same count.

Since $j_1<j_2$, $w_1$ is a strict prefix of $w_2$, i.e., $w_2=w_1\cdot 0^{j_2-j_1}$.
Assume by contradiction that after processing $w_1,w_2$ \alg{} reaches the same memory configuration.
Since \alg{} is deterministic, this means that it must reach the same configuration after seeing $w_1\cdot 0^{z(j_2-j_1)}$ for any integer $z$.
By choosing $z=W(1+\tau)$, we get that the algorithm reaches this configuration once again while the entire window consists of zeros.
This is a contradiction since $\sigma_1,\sigma_2\neq 0$, and the algorithm cannot answer both $w_1$ and $w_1\cdot 0^{z(j_2-j_1)}$ correctly.
\end{proof}

\section{Proof of Theorem~\ref{thm:addLB} }\label{app:add-lb-proof}
Before we prove Thorem~\ref{thm:addLB}, we start with a simpler lower bound.
\begin{lemma}\label{lem:logweps^1}
Let $\eps < 1/4$. Any deterministic algorithm that solves the \ADDI{} problem must use at least ${\log (W/\eps)}-O(1)$ bits.
\end{lemma}
\begin{proof}
Denote by $rep(x)\triangleq(x\mod R)\cdot R^{\floor{x/R}}$ a sequence in $\set{\sigma R^*\mid \sigma\in[R]}$ whose sum is $x$.
Next, consider the following languages:
$${L_{A_1}}\triangleq \set{rep(k\cdot 2\bserror)\mid k\in[\floor {1/4\eps}]\setminus\set{0}};\qquad \overline{L_{A_1}}
\triangleq 0^{W+\wt}\cdot L_{A_1}\cdot \set{0^{q}\mid q\in\left[\floor{W/2}\right]}.$$
First, notice that $|L_{A_1}|=\floor {1/4\eps}$ and that all words in $L_{A_1}$ have length of at most $W/2$.
This means that $|\overline{L_{A_1}}| = \floor {1/4\eps}\floor{W/2+1} > \floor{W/8\eps}$.

We now show that every word in $\overline{L_{A_1}}$ must have a dedicated memory configuration, thereby implying a $\clog{\floor{W/8\eps}}$ bits bound.
Let $w_1=0^{W+\wt}\cdot x_1\cdot0^{q_1}$ and $w_2=0^{W+\wt}\cdot x_2\cdot0^{q_2}$ be two distinct words in $\overline{L_{A_1}}$ such that $x_1,x_2\in L_{A_1}$ and $q_1,q_2\in\floor{W/2}$.
If $x_1\neq x_2$, then their most recent $W$ elements differ by more than $2RW\epsilon$ and there is no output that is correct for both.
Note that the slack of both $w_1$ and $w_2$ is all zeros. Hence, $w_1$ and $w_2$ require different memory configurations.

Assume that $x_1=x_2$ and that by contradiction both $w_1$ and $w_2$ reached the same memory configuration.
Since $w_1 \neq w_2$ and $x_1=x_2$, then $q_1 \neq q_2$ and without loss of generality $q_1<q_2$.
This implies that $w_1$ is a prefix of $w_2$ so that $w_2=w_1\cdot 0^{q_2-q_1}$.
Thus, $\mathbb{A}$ enters the shared configuration after reading $w_1$ and revisits it after reading $0^{q_2-q_1}$.
$\mathbb{A}$ is a deterministic algorithm and therefore it reaches the same configuration also for the following word: $w_1\cdot 0^{(W+\wt)(q_2-q_1)}$.
In that word, the last $W+\wt$ elements are all zeros while the sum of the last $W$ elements in $w_1$ is at least $2RW\epsilon$.
Hence, there is no return value that is correct for both $w_1$ and $w_1\cdot 0^{(W+\wt)(q_2-q_1)}$.

\end{proof}

We are now ready to prove Theorem~\ref{thm:addLB}. The theorem says that for $\eps < 1/4$, any deterministic algorithm \alg{} that solves the \ADDI{} problem requires\\ $\max\set{{\log (W/\eps)}-O(1),\ceil{\ceil{\oneOverT/2}{\log \floor {{\tau/2\eps+1}}}}}$~bits.
\begin{proof}
Lemma~\ref{lem:logweps^1} shows that \alg{} must use at least ${\log (W/\eps)}-O(1)$ bits. Given $x\in[RW\tau]$, we denote by $rep(x)\triangleq0^{2W\tau - \floor{x/R} - 1}\cdot(x\mod R)\cdot R^{\floor{x/R}}$ a sequence of the following form: $\set{0^{W\tau+i}\sigma R^{W\tau-i-1}\mid i\in[W\tau-1], \sigma\in[R]}$ whose sum is $x$.
We consider the following languages:
$${L_{A_2}}\triangleq \set{rep(k\cdot 2\bserror)\mid k\in[\floor {\tau/2\eps}]};\qquad
\overline{L_{A_2}}
\triangleq \set{w_1\cdot w_2\cdots w_{\ceil{\oneOverT/2}}\mid \forall i: w_i\in L_{A_2}}.$$

Our goal is to show that no two words in $\overline{L_{A_2}}$ have the same memory configuration.
Let $S_1, S_2 \in \overline{L_{A_2}}$ so that $S_1 \neq S_2$.
Denote
 $S_1\triangleq w_{1,1}\cdot w_{2,1}\cdots w_{\floor{\oneOverT/2},1}$ and $S_2\triangleq  w_{1,2}\cdot w_{2,2}\cdots w_{\floor{\oneOverT/2},2}$, while $\forall i:w_{i,1},w_{i,2}\in L_{A_2}$.
We denote $\chi\triangleq \max\set{i\mid w_{i,1}\neq w_{i,2}}$ -- the last place in which $S_1$ differs from $S_2$.

Next, consider the following sequences: $S^*_1 = S_1 \cdot 0^{2\wt(\chi-1/2)}$ and $S^*_2 = S_2 \cdot 0^{2\wt(\chi-1/2)}$.
The last $W+\wt$ elements in $S^*_1$ are $w_{\chi,1}\cdots w_{\ceil{\oneOverT/2},1}\cdot 0^{2W(\chi-1/2) \tau}$ and in $S^*_2$ $w_{\chi,2}\cdots w_{\ceil{\oneOverT/2},2}\cdot 0^{2W(\chi-1/2) \tau}$.
Additionally, the $W\tau$ elements slack in both $S^*_1$ and $S^*_2$ are all zeros.
Now, since the sum of $w_{\chi,1}$ and $w_{\chi,2}$ must differ by at least $2RW\eps$, no number can approximate both with less than $\bserror$ error.
\end{proof}

\section{Proof of Lemma~\ref{lem:mult-lb1}}\label{app:multLBLemma}
Before we prove Lemma~\ref{lem:mult-lb1} we first give a bound on an integer set that will serve us in the main lemma.
\begin{lemma*}\label{lem:multSetSize}
Consider the integer set $I_{M}\triangleq\set{a_n\mid a_n\le R\floor{W/2}}$, where the integers $\{a_i\}$ are taken from the following sequence: $a_1=1, \forall n>1: a_n=\ceil{(1+\epsilon)a_{n-1}}$.
The cardinality of $I_{M}$ satisfies
$\left|I_{M}\right| \ge \lnp 2\oneOverE\logp{RW\eps/2}-O(1)$.

\end{lemma*}
\begin{proof}
We first
show an upper bound on $a_n$  $\forall n\in\mathbb N: a_n\le \epsilon^{-1}\cdotpa{(1+\epsilon)^n-1}$.
\begin{itemize}
\item \textbf{Basis}: for $n=1$, we have
$a_1 = 1 =\eps^{-1}\cdotpa{(1+\epsilon)-1}$.
\item \textbf{Hypothesis:} $a_{n-1}\le \epsilon^{-1}\cdotpa{(1+\epsilon)^{n-1}-1}$.
\item \textbf{Step:} For $n>1$, we bound $a_n$ as follows:
\small
\begin{align*}
a_n = \ceil{(1+\epsilon)a_{n-1}} & \le \ceil{(1+\epsilon)\epsilon^{-1}\cdotpa{(1+\epsilon)^{n-1}-1}}\\
&=\ceil{\epsilon^{-1}(1+\epsilon)^n-\epsilon^{-1}-1}<\epsilon^{-1}\cdotpa{(1+\epsilon)^{n}-1}.
\end{align*}
\normalfont
\end{itemize}
Next, notice that this implies that $|I_M|\ge\arg\max\set{n\mid \epsilon^{-1}(1+\epsilon)^n\le R\floor{W/2}}$.
Finally, we get a lower bound of $n=\floor{\log_{1+\epsilon}(R\floor{W/2}\eps)}=\frac{\lnp{RW\eps/2}}{\ln(1+\eps)}-O(1) < \lnp2\epsilon^{-1}\allowbreak{\logp{RW\eps/2}}-O(1)$, where the last inequality follows from the Taylor expansion of $\lnp{1+\eps}$.\qedhere
\end{proof}
We are now ready to prove Lemma~\ref{lem:mult-lb1} using the lemma above.
\begin{lemma*}
For $\eps < 1/4$, any deterministic algorithm \alg{} for the \MULT{} problem requires at least ${\log (W/\eps)+\log\logp {\bserror}}-O(1)$ memory bits.
\end{lemma*}
\begin{proof}
		We show a language $\overline{L_M}$ for which every two words must reach a unique memory configuration, thus implying a $\ceil{\log{|\overline{L_{M}}|}}$ bits lower bound.
		We denote by  $rep(x)\triangleq(x\mod R)\cdot R^{\floor{x/R}}$ a sequence in $\set{\sigma R^*\mid \sigma\in[R]}$ that has a sum of $x$. We define $\overline{L_M}$ as~follows:
		$$\overline{L_{M}}\triangleq \set{0^{W+\wt}rep(x)0^j\mid j\in[\lfloor W/2\rfloor], x\in I_{M}}.
		$$
		Notice that $|\overline{L_{M}}| = |{I_{M}}|\cdotpa{\floor{W/2}+1}$ and according to Lemma~\ref{lem:multSetSize} we have
		$$\scriptsize \log|\overline{L_{M}}|\ge \logp{\lnp 2\oneOverE\logp{RW\eps/2}-O(1)} + \logw - 1 = \log (W/\eps)+\log\logp {\bserror}-O(1).$$
		Consider two words $w_1=0^{W+\wt}rep(x_1)0^{j_1}$ and $w_2=0^{W+\wt}rep(x_2)0^{j_2}$ in $\overline{L_{M}}$, such that $x_1,x_2\in I_M$.
		Notice that every two distinct numbers $z<q\in I_M$ satisfy $z\le q/(1+\eps)$.
		Since $w_1,w_2$ are preceded with a sequence of $\wt$ zeros, no answer correctly satisfies the requirements for both of them.
		Thus, if $x_1\neq x_2$, \alg{} must reach a different configuration after processing $w_1$ than it reaches when reading $w_2$.
		Next, assume that $x_1=x_2$, and that without loss of generality $j_1<j_2$ (i.e., we have $w_2=w_1\cdot 0^{j_2-j_1}$).
		If the algorithm reached some configuration $\mathfrak c$ after reading $w_1$ and then returned to it after reading the additional $j_2-j_1$ zeros of $w_2$, it must get back to $\mathfrak c$ when reading $w_1\cdot 0^{W(1+\tau)(j_2-j_1)}$.
		Notice that the sum of $w_1$ is $x_1$ and the current window sums to zero.
		Thus, no single $\widehat{S}$ value would satisfy both sequences.
		Hence, \alg{} must reach a different configuration than $\mathfrak c$ when seeing~$w_2$.
\end{proof}

\section{Tighter Analysis of the \EXACT{} Algorithm}
\label{app:tighter}

\begin{theorem}
Consider a stream where $R=O(1)$ and $\tau=1$. There exists a \EXACT{} algorithm that uses $1.5\mathcal B+O(1)$ bits, where $\mathcal{B}$ is the lower~bound.
\end{theorem}
\begin{proof}
Our method here is similar to Algorithm~\ref{alg:exact}, but the constant $\tau$ value allows us to compute $\sum_{i=1}^{\oneOverT} b_i$ in $O(1)$ without tracking it in $B$. Thus, our algorithm only requires $3\logw+O(1)$ bits, while Theorem~\ref{thm:exactLB} gives a lower bound of $\floor{2\logw}$.
\end{proof}

\section{Correctness Proof of Algorithm~\ref{alg:exact} }\label{app:exact-alg-prof}
\begin{theorem}
	Algorithm~\ref{alg:exact} solves the \EXACT{} problem.
\end{theorem}
\begin{proof}
First, notice that $c$ is always in the range $\xrange{\wt}$ and thus the slack size is as needed. Next, assume that the algorithm input was $x_1,x_2,\ldots,x_t$; since stream is broken into blocks of size $\wt$, where $c$ is the offset within the current block, we have that $c= t \mod \wt$. Further, the last $c$ elements are $x_{t-c+1},\ldots,x_t$ and the preceding $W$ elements are $x_{t-c-W+1},\ldots,x_{t-c}$. Finally, the algorithm always keep the sum of the $W$ elements before the current block in $B$, and the sum of the current block in $y$; thus, by returning $\widehat{S}=B+y$ we get exactly the sum of the last $W+c$ elements.
\end{proof}
\section{Correctness Proof of Algorithm~\ref{alg:addi}}
\label{app:AddAlgProof}
\begin{theorem*}
Algorithm~\ref{alg:addi} solves the \ADDI{} problem.
\end{theorem*}
\begin{proof}
First, observe that at all times $c\in[\wt-1]$ as needed.
Denote the stream by $S=x_1,x_2,\ldots x_{t+c}$, such that $c$ represents the number of elements within the current block.
Our goal is to show that Algorithm~\ref{alg:addi} provides a $\bserror$ approximation to the sum of the last $W+c$ elements ($x_{t-W+1}\ldots x_{t+c}$). That is, the quantity we approximate is
\ifdefined\TWELVEPAGER	
	$S\triangleq\sum_{\ell=1}^{W+c}x_{t-W+\ell}.$
\else
	$$S\triangleq\sum_{\ell=1}^{W+c}x_{t-W+\ell}.$$
\fi

 For any $\ell\in[t+c]$, we use $y_{\ell}$ to denote the value of the variable $y$ after the $\ell$'th item was added.
Note that within a block, $y$ simply sums the rounded scaled inputs; whenever a block ends, we reduce the value of $y$ by $\wt\cdot\text{Round}_{\bsReminderFractionBits_2}(\frac{y}{{\wt}})$,
but make up for it by setting $b_i$. Further, when processing $x_{t-W+1}\ldots x_{t+c}$, we replace all of the values of $b$ that were determined before the last $W+c$ elements, and none of the set value leaves the window by time $t+c$. That is, $i$ reaches every value in $[\oneOverT]$ exactly once throughout the last $W+c$ updates.
This gives us the following equality:
$
\resizeIfTWELVEPAGER
{y_{t-W} + \sum_{i=1}^{W+c} x'_{t-W+i} = \wt\sum_{i=1}^{\oneOverT} b_i + y_{t+c}  = \wt \cdot B + y_{t+c}.}
$\\
Thus, we can express the algorithm's estimate of the sum value $\widehat{S} \triangleq R\cdotpa{\wt \cdot B + y_{t+c}}$ as:
\begin{align}
\resizeIfTWELVEPAGER{\widehat{S} = R\cdotpa{y_{t-W} + \sum_{i=1}^{W+c} x'_{t-W+i}}}.\label{eq:slackyEstimation}
\end{align}
Next, notice that since $\forall \ell: x_\ell' = \text{Round}_{\bsReminderFractionBits_1}(\frac{x_\ell}{R})$, we have $|x_\ell'-\frac{x_\ell}{R}|\le 2^{-\bsReminderFractionBits_1-1}$ and thus:
\begin{align}
 {
	\left|\sum_{\ell=1}^{W+c} x'_{t-W+\ell} -  \frac{1}{R}\sum_{\ell=1}^{W+c} x_{t-W+\ell}\right| = \left|\sum_{\ell=1}^{W+c} x'_{t-W+\ell} -  \frac{1}{R}S\right| \le (W+c)\cdot 2^{-\bsReminderFractionBits_1 - 1}.
	} \label{eq:slackyFirstRounding}
\end{align}
Also, since we assumed that $x_{t-W}$ was the last of a $\wt$-sized block, we know that the value of $y$ is bounded, and specifically:
\begin{align}
|y_{t-W}| \le \wt\cdot 2^{-\bsReminderFractionBits_2 -1}.\label{eq:slackyYBound}
\end{align}
Plugging \eqref{eq:slackyFirstRounding} and \eqref{eq:slackyYBound} into \eqref{eq:slackyEstimation}, we get a bound on the error:
\begin{align*}
\resizeIfTWELVEPAGER{\mathcal E \triangleq |S - \widehat{S}| \le R\cdotpa{\wt\cdot 2^{-\bsReminderFractionBits_2 -1} + (W+c)\cdot 2^{-\bsReminderFractionBits_1 - 1}}< RW\cdotpa{\tau 2^{-\bsReminderFractionBits_2 -1} + 2^{-\bsReminderFractionBits_1}}}.
\end{align*}
Thus, since $\bsReminderFractionBits_1 = \ceil{\log\oneOverE}+1$, $\bsReminderFractionBits_2 = \ceil{\log\frac{\tau}{\eps}}$, we get the desired $\mathcal E < \bserror$ error~bound and conclude that Algorithm~\ref{alg:addi} solves \ADDI{}.
\end{proof} 
\section{Analysis of Algorithm~\ref{alg:mult}}\label{app:mult-lb-proof}

We now analyze the memory requirement of our algorithm.
\begin{theorem}
Algorithm~\ref{alg:mult} requires $O\parentheses{\oneOverT\parentheses{\log\logp {RW\tau} + \log{\oneOverE}}+\log(RW)}$ bits.
\end{theorem}
\begin{proof}
Since $\rho\in\parentheses{\set{-\infty}\cup\set{0,1,\ldots ,\floor{\log_{\logBase}RW\tau}}}$, it can be represented using \\$\ceil{\logp{\log_{\logBase}(RW\tau) + 2}}$ bits. Next, notice that this satisfies:
\begin{align*}
\footnotesize
\ceil{\logp{\log_{\logBase}(RW\tau) + 2}}
&= \logp{\frac{\log RW\tau}{\log \logBase}} + O(1)\\
&= \logp{\log RW\tau} - \log\log \logBase + O(1)\\
&= \logp{\log RW\tau} - \logp{\frac{\ln \logBase}{\ln 2}} + O(1)\\
&= \logp{\log RW\tau} - \logp{\frac{\eps/2}{\ln 2}+O(\eps^2)} + O(1)\\
&= \logp{\log RW\tau} + \log{\oneOverE} + O(1).
\end{align*}
Each counter in $b$ now is assigned a $\rho$ value and thus the overall space consumption of $b$ is $\oneOverT\parentheses{\logp{\log RW\tau} + \log{\oneOverE} + O(1)}$ bits.
Our $B$ variable sums the values rather than the exact sum of each block and is a fixed point variable. We use $\clogp{RW\tau+1}$ bits for its integral part and another $\clogp{4\oneOverE}$ for its fractional part. 
Thus, the total number of bits required by the algorithm~is 
\begin{multline*}
\oneOverT(\log{\logp {RW\tau}} + \log{\oneOverE}+O(1)) + O(\log(RW))\\
 = O\parentheses{\oneOverT\parentheses{\log\logp {RW\tau} + \log{\oneOverE}}+\log(RW)}.\qquad\qedhere
\end{multline*}
\end{proof}
We thus get that Algorithm~\ref{alg:mult} is optimal under some conditions; notice that this includes constant $\tau$ values.
\begin{theorem}
Let $\mathcal B = \Omega\big(\log (W/\eps)\allowbreak+ \oneOverT\parentheses{\logp{\tau/\eps} + \log\logp{RW}}\big)$ be the \MULT{} lower bound showed in Theorem~\ref{thm:multLB}.
Then for $\tau^{-1}=O\parentheses{{\log W}}$ and $R=W^{O(1)}$, Algorithm~\ref{alg:mult} uses $O(\mathcal{B})$ bits.
\end{theorem}

Next, we prove that Algorithm~\ref{alg:mult} solves the problem. 
Recall that for a real number $x$, we denote $\roundDown x\triangleq \floor{x\cdot k}/k$, for $k\triangleq \ceil{4\over \epsilon}$.
\begin{lemma}\label{lem:mult-alg-roundDown}
Let $\epsilon\le 1/2$; for every $y\in \mathbb R^+$ such that $y>0$ the following inequality holds:
$$ \frac{y}{1+\eps} < \roundDown{\logBase^{\floor{\log_{\logBase}y}}}\le y$$
\end{lemma}
\begin{proof}
Observe that we have
$$\roundDown{\logBase^{\floor{\log_{\logBase}y}}} > \roundDown{\logBase^{{\parentheses{\log_{\logBase}y}-1}}} = \roundDown{\frac{y}{\logBase}} \ge \frac{y}{\logBase}-\eps/4.$$
Finally, since $y\ge 1$ and $\epsilon\le 1/2$, we get that $\frac{y}{\logBase}-\eps/4 \ge \frac{y}{1+\eps}$.
\end{proof}
\begin{theorem}
	Algorithm~\ref{alg:mult}  solves the \MULT{} problem.
\end{theorem}
\begin{proof}
	First, observe that at all times $c\in[\wt-1]$ as needed.
	Denote the stream by $S=x_1,x_2,\ldots x_{t+c}$, such that $c$ represents the number of elements within the current block.
	We will show that Algorithm~\ref{alg:mult} provides a $(1+\eps)$ approximation to the sum of the last $W+c$ elements ($x_{t-W+1}\ldots x_{t+c}$). That is, the quantity we approximate is
	$$S\triangleq\sum_{\ell=1}^{W+c}x_{t-W+\ell}.$$ 
	Next, since we reset $\remainder$ in every block (Line~\ref{line:mult-y-reset}), we have that at the time of the query, 
$ y = \sum_{\ell = t+1}^{t+c} x_\ell.$
	Further, since every block is summed individually and then rounded in Line~\ref{line:mult-rho}. We assume without loss of generality that $i=0$ at the time of the query, 
	and $\forall \jmath \in\brackets{\oneOverT-1}$ we denote $S_\jmath\triangleq \sum_{\ell = t-\wt(\oneOverT-\jmath)+1}^{t-\wt(\oneOverT-\jmath-1)} x_\ell$ -- the sum of the elements in block $\jmath$.
	Therefore we  have:
	\begin{align}
	\forall \jmath \in\brackets{\oneOverT-1}: b_\jmath = \floor{\log_{\logBase} S_\jmath}.\label{eq1}
	\end{align}
	That is each of the $\oneOverT$ blocks that precede the last $c$ elements is summed exactly (Line~\ref{line:mult-exact-inblock-summing}) and then we store its $\logBase$-based log in $b$.
	Next, the $B$ variable stores the approximated values rounded down (Line~\ref{line:mult-B-summing}), i.e., 
	\begin{align}
	B = \sum_{\jmath = 0}^{\oneOverT-1} \roundDown{{\logBase^{b_\jmath}}}.\label{eq2}
	\end{align}	
	
	Now, Lemma~\ref{lem:mult-alg-roundDown} implies that
	\begin{align}
	\forall \jmath \in\brackets{\oneOverT-1}: S_\jmath>0 \implies \frac{S_\jmath}{1+\eps} < \roundDown{\logBase^{\floor{\log_{\logBase}S_\jmath}}}\le S_\jmath.\label{eq3}
	\end{align}
	We then get
	\begin{multline}
		S=\sum_{\ell=1}^{W+c}x_{t-W+\ell} = \sum_{\ell=1}^{W}x_{t-W+\ell} + \sum_{\ell=W+1}^{W+c}x_{t-W+\ell} = \sum_{\jmath = 0}^{\oneOverT-1} S_\jmath + y
		= \sum_{\substack{\jmath\in\brackets{\oneOverT-1}:\\S_\jmath> 0}}^{\oneOverT-1} S_\jmath + y. \label{eq4}
	\end{multline}
	Also, note that if some $S_j=0$, then we defined $b_\jmath$ as $-\infty$ and $\logBase^{b_\jmath}$ as $0$.	
	This means that if $S=0$, then $\forall \jmath \in\brackets{\oneOverT-1}: S_\jmath=0, y=0$ and thus $\widehat{S}=0$. Hereafter, assume that $S\neq 0$.
	Next, we plug equations \eqref{eq1},\eqref{eq2},\eqref{eq3} into \eqref{eq4} to get:
	\begin{align*}
	S= \sum_{\substack{\jmath\in\brackets{\oneOverT-1}:\\S_\jmath> 0}} S_\jmath + y \ge \sum_{\substack{\jmath\in\brackets{\oneOverT-1}:\\S_\jmath> 0}} \roundDown{{\logBase^{b_\jmath}}} + y = B + y = \widehat{S}.
	\end{align*}	
	Similarly, we bound $\widehat{S}$ from below as follows:
	\begin{align*}
	\frac{S}{1+\eps}= \frac{\sum\limits_{\substack{\jmath\in\brackets{\oneOverT-1}:\\S_\jmath> 0}} S_\jmath + y}{1+\eps} < \sum_{\substack{\jmath\in\brackets{\oneOverT-1}:\\S_\jmath> 0}} \roundDown{{\logBase^{b_\jmath}}} + y = B + y = \widehat{S}.
	\end{align*}	
	We showed that in all cases where $S\neq 0$ we have $\frac{S}{1+\eps} < \widehat{S}\le S$ if $S>0$, thereby proving the theorem.
\end{proof}
%


\section{Analysis of Algorithm~\ref{alg:mult2}}\label{apx:mult2}
In order to analyze Algorithm~\ref{alg:mult2}, we first note that by rounding down a real number to use $\ceil{\log\wt}$ bits, as in the $\remRoundDown{\cdot}$ operator, we introduce a rounding error of at most $-1/\wt$.
\begin{observation}
For any $\alpha\in\mathbb R^+: \alpha-1/\wt<\remRoundDown{\alpha}\le \alpha$.
\end{observation}
Our approach in the analysis of Algorithm~\ref{alg:mult2} is as follows:
\begin{enumerate}
\item We start with Lemma~\ref{lem:multYVal} that shows that $\altLogBase^y$ is a $\altLogBase$ multiplicative approximation to the block's sum.
\item Next, Lemma~\ref{lem:bVal} shows that we do not lose much by taking $\floor y$ into our cyclic buffer $b$ (rather than $y$ itself). This allows us to reduce the memory requirement at the expense of slightly increasing the error. Specifically, we show that $\altLogBase^{\floor y}$ is a $(1+\eps)$ multiplicative approximation of the sum.
\item Then we proceed with Lemma~\ref{lem:mult2YBound} that shows a $O(\oneOverE\logp{RW\tau})$ bound on $y$. This allows us to bound the number of bits needed for the representation of its integral part.
\item Lemma~\ref{lem:mult2Mem} analyzes the overall space requirement of Algorithm~\ref{alg:mult2}.
\item Next, Theorem~\ref{alg:mult2correct} shows we indeed solve \MULT{}.
\item Finally, Corollary~\ref{cor:mult2} concludes the optimality for constant $\tau$.
\end{enumerate}

\begin{lemma}\label{lem:multYVal}
Let $x_1,\ldots x_{\wt}$ be the elements of a block summed in $\remainder$, then $\sum_{i=1}^{\wt} x_i = 0\implies y=-\infty$ and otherwise:
$$
	\frac{\sum_{i=1}^{\wt} x_i}{\altLogBase} < \altLogBase^y \le \sum_{i=1}^{\wt} x_i.
$$
\end{lemma}
\begin{proof}
We prove the lemma by showing that $\forall n\in\frange{\wt}$, after summing $x_1,\ldots, x_n$ we have that if $\sum_{i=1}^{n} x_i=0$ then $y=-\infty$ and otherwise:
$$
\frac{\sum_{i=1}^{n} x_i}{\altLogBase^{n/\wt}} < \altLogBase^y \le \sum_{i=1}^{n} x_i .
$$
The proof is done by induction where we denote by $y_i$ the value of $y$ after summing $x_1,\ldots x_{i}$. 
\begin{itemize}
\item \textbf{Basis:} $n=0$.\\
Here we simply have $y=-\infty$ and the claim holds.
\item \textbf{Induction hypothesis:} let $0<n<\wt$ then if $\sum_{i=1}^{n} x_i=0$ then $y_n=-\infty$ and otherwise:
$$
\frac{\sum_{i=1}^{n} x_i}{\altLogBase^{n/\wt}} < \altLogBase^{y_n} \le \sum_{i=1}^{n} x_i .
$$
\item \textbf{Induction step:} let $y_{n+1} \triangleq \remRoundDown{\log_{\altLogBase}\parentheses{{x_{n+1}+\altLogBase^{\remainder_n}}}}$.
\end{itemize}
We first consider the case where $x_1=\ldots=x_{n+1}=0$. In this case we have $y_n=-\infty$ according to the induction hypothesis.
Thus, we get
$$
y_{n+1} = \remRoundDown{\log_{\altLogBase}\parentheses{{x_{n+1}+\altLogBase^{\remainder_n}}}} =\remRoundDown{\log_{\altLogBase}(0)}=\remRoundDown{-\infty}= -\infty.
$$
Next, consider the case where $x_1=\ldots=x_{n}=0$ but $x_{n+1}>0$. This also gives us $y_n=-\infty$ and thus:
\begin{align*}
&\altLogBase^{y_{n+1}} = \altLogBase^{\remRoundDown{\log_{\altLogBase}\parentheses{{x_{n+1}+\altLogBase^{\remainder_n}}}}} =
 \altLogBase^{\remRoundDown{\log_{\altLogBase}{{x_{n+1}}}}} \le x_{n+1} = \sum_{i=1}^{n+1} x_i ,
\end{align*}
and
\begin{align*}
&\altLogBase^{y_{n+1}} = \altLogBase^{\remRoundDown{\log_{\altLogBase}{{x_{n+1}}}}}
\ge x_{n+1}\altLogBase^{-1/\wt} \ge \frac{\sum_{i=1}^{n+1} x_i}{\altLogBase^{n/\wt}}.
\end{align*}

Finally, consider the case where ${\sum_{i=1}^{n} x_i}>0$, and according to the induction hypothesis:
$$
\frac{\sum_{i=1}^{n} x_i}{\altLogBase^{n/\wt}} < \altLogBase^{y_n} \le \sum_{i=1}^{n} x_i .
$$
Thus, we bound $\altLogBase^{y_{n+1}}$ as follows:
\begin{align*}
\altLogBase^{y_{n+1}} &= \altLogBase^{\remRoundDown{\log_{\altLogBase}\parentheses{x_{n+1}+\altLogBase^{\remainder_n}}}} \le
\altLogBase^{\remRoundDown{\log_{\altLogBase}\parentheses{x_{n+1}+\sum_{i=1}^{n} x_i}}} \\
&=\altLogBase^{\remRoundDown{\log_{\altLogBase}{\sum_{i=1}^{n+1} x_i}}} \le \sum_{i=1}^{n+1} x_i.
\end{align*}
Similarly, we can bound it from below:
\begin{align*}
\altLogBase^{y_{n+1}} &= \altLogBase^{\remRoundDown{\log_{\altLogBase}\parentheses{x_{n+1}+\altLogBase^{\remainder_n}}}} >
\altLogBase^{\remRoundDown{\log_{\altLogBase}\parentheses{x_{n+1}+\frac{\sum_{i=1}^{n} x_i}{\altLogBase^{n/\wt}}}}}\\
&> \altLogBase^{\remRoundDown{\log_{\altLogBase}\parentheses{\frac{\sum_{i=1}^{n+1} x_i}{\altLogBase^{n/\wt}}}}}\\
&\ge \parentheses{\frac{\sum_{i=1}^{n+1} x_i}{\altLogBase^{n/\wt}}}\altLogBase^{-1/\wt} = \frac{\sum_{i=1}^{n+1} x_i}{\altLogBase^{(n+1)/\wt}}.
\end{align*}
\end{proof}
\begin{lemma}\label{lem:bVal}
Let $x_1,\ldots x_{\wt}$ be the elements of a block summed in $\remainder$, then  $\sum_{i=1}^{\wt} x_i=0\implies {{\altLogBase^{\floor\remainder}}} = 0$ and otherwise
$$
\frac{\sum_{i=1}^{\wt} x_i}{1+\epsilon} < {{\altLogBase^{\floor\remainder}}} \le \sum_{i=1}^{\wt} x_i.
$$
\end{lemma}
\begin{proof}
First, notice that Lemma~\ref{lem:multYVal} implies that if $\sum_{i=1}^{\wt} x_i=0$ then $y=-\infty$ and thus ${{\altLogBase^{\floor\remainder}}} = 0$.
If the sum is non-zero, the lemma implies:
\begin{align}
	\frac{\sum_{i=1}^{\wt} x_i}{\altLogBase} < \altLogBase^y \le \sum_{i=1}^{\wt} x_i.
\end{align}
Thus, we have that:
\begin{align*}
{{\altLogBase^{\floor\remainder}}} \le  {{\altLogBase^{\remainder}}} \le \sum_{i=1}^{\wt} x_i.
\end{align*}
From below, we can bound it as follows:
\begin{align*}
\altLogBase^{\floor\remainder} > \frac{\altLogBase^{\remainder}}{\altLogBase} > \frac{\sum_{i=1}^{\wt} x_i}{\altLogBase^2} > \frac{\sum_{i=1}^{\wt} x_i}{1+\epsilon},
\end{align*}
where the last inequality holds as $\epsilon\le 1/2$.
\end{proof}
We now bound the value of $\remainder$ in order to compute the memory requirements of Algorithm~\ref{alg:mult2}
\begin{lemma}\label{lem:mult2YBound}
For any $n\in\frange{\wt}$, let $x_1,\ldots x_{n}$  be elements of a block summed in $\remainder$ (Line~\ref{line:mult-exact-inblock-summing2}), then $y = O(\oneOverE\logp{RW\tau})$.
\end{lemma}
\begin{proof}
According to Lemma~\ref{lem:multYVal} we have that $\altLogBase^{y} \le \sum_{i=1}^{n} x_i \le R\wt$. Thus we can get an upper bound on $\remainder$'s value:
$$
y \le \log_{\altLogBase}(R\wt) = \frac{\lnp{R\wt}}{\ln\altLogBase} = O(\oneOverE\logp{RW\tau}).
$$
\end{proof}
We are now ready to compute the space requirement of Algorithm~\ref{alg:mult2}.
\begin{lemma}\label{lem:mult2Mem}
Algorithm~\ref{alg:mult2} uses $O\parentheses{\oneOverT\parentheses{\log\logp {RW\tau} + \log{\oneOverE}}+\log W}$ space.
\end{lemma}
\begin{proof}
The algorithm uses four variables: 
\begin{itemize}
	\item $\remainder$ -- a fixed point variable with $O(\log\wt)$ bits for its fractional part. The integral part of $\remainder$, according to Lemma~\ref{lem:mult2YBound}, can be represented using $O(\log\oneOverE+\log\logp{R\wt})$.
	\item $\bitarray$ -- a cyclic array with $\oneOverT$ entries, each of which can be represented using $O(\log\oneOverE+\log\logp{R\wt})$ bits per Line~\ref{line:floorY2}.
	\item $\currentBlockIndex$ -- tracks the index of the current block and thus has $\oneOverT$ possible values and can be represented using $O(\log\oneOverT)$ bits.
	\item $\blockOffset$	-- the offset within the current block; it is represented using $O(\log\wt)$ bits.
\end{itemize}
Overall, we get that the memory requirement is as stated.
\end{proof}
We now prove the correctness of Algorithm~\ref{alg:mult2}.
\begin{theorem}
\label{alg:mult2correct}
Algorithm~\ref{alg:mult2} processes elements in constant time, answer queries in $O(\oneOverT)$, uses $O\parentheses{\oneOverT\parentheses{\log\logp {RW\tau} + \log{\oneOverE}}+\log W}$ space and solves the \MULT{} problem.
\end{theorem}
\begin{proof}
Let $x_1,\ldots x_{W+c}$ be the elements we are trying to approximate.
Observe that $\forall \jmath\in[\oneOverT-1]$, the elements $x_{\wt\cdot \jmath + 1, \ldots \wt\cdot(\jmath+1)}$ are summed together (Line~\ref{line:mult-exact-inblock-summing2}) and then stored in some $b_i$ (Line~\ref{line:floorY2}). This means that, according to Lemma~\ref{lem:bVal}, $\altLogBase^{b_i}$ approximates their sum up to a multiplicative error of $1+\epsilon$. Thus, we have that:
$$
\frac{\sum_{i=1}^{W} x_i}{1+\epsilon} < \sum_{i=1}^{\oneOverT} \altLogBase^{b_i} \le \sum_{i=1}^{W} x_i.
$$
Finally, according to Lemma~\ref{lem:multYVal}, we have that $y$ approximates the last $c$ elements and specifically:
$$
\frac{\sum_{i=W+1}^{W+c} x_i}{1+\epsilon} < \frac{\sum_{i=W+1}^{W+c} x_i}{\altLogBase} < \altLogBase^{y} \le \sum_{i=W+1}^{W+c} x_i.
$$
This allow us to conclude that:
$
\frac{\sum_{i=1}^{W+c} x_i}{1+\eps} < \altLogBase^{\remainder} + \sum_{i=1}^{\oneOverT} \altLogBase^{b_i} \le \sum_{i=1}^{W+c} x_i.
$
\end{proof}

\begin{corollary}\label{cor:mult2}
For $\tau=\Theta(1)$, Algorithm~\ref{alg:mult2} answer queries in $O(1)$ time, uses $O(\log(W/\eps)+\log\log R)$ bits, and is asymptotically optimal.
\end{corollary}
\section{Proof of Theorem~\ref{thm:max-thm}}\label{app:max-thm}
\begin{theorem*}
	Tracking the maximum over a slack window deterministically requires $O\parentheses{\oneOverT\log R}$ and $\Omegap{\oneOverT\log R\tau}$ bits.
\end{theorem*}
\begin{proof}
	The algorithm we propose is quite simple -- compute the maximum over each $\wt$-sized block and keep a cyclic buffer the last $\tau^{-1}$ blocks' maxima. Then, we can compute the maximum of the cyclic buffer at query time to get the maximal value in the slack window. 
	For a lower bound, consider the following language:
	{
		$$L_{\max}\triangleq 
		\set{0^{W-2\wt\floor{1/2\tau}}\sigma_1^{2\wt}\sigma_2^{2\wt}\ldots\sigma_{\floor{1/2\tau}}^{2\wt}\mid \forall i:\sigma\in[R]\wedge (\sigma_1\le\sigma_2\le\ldots\le\sigma_{\floor{1/2\tau}})}.
		$$
	}
	We now claim that each pair of distinct words $w_1,w_2\in L_{\max}$ must lead the algorithm into a distinct memory configuration.
	Denote $w_1 = 0^{W-2\wt\floor{1/2\tau}}\sigma_{1,1}^{2\wt}\ldots\sigma_{1,\floor{1/2\tau}}^{2\wt}$,
	$w_2 = 0^{W-2\wt\floor{1/2\tau}}\sigma_{2,1}^{2\wt}\ldots\sigma_{2,\floor{1/2\tau}}^{2\wt}$ and let $t\triangleq\max\set{i\mid \sigma_{1,i}\neq \sigma_{2,i}}$. Without loss of generality assume that $\sigma{1,t}>\sigma_{2,t}$. If $w_1$ and $w_2$ lead the algorithm into the same memory configuration, then it (being deterministic) must reach the same configuration again and provide the same output for $w_1\cdot 0^{W-t\cdot 2\wt}$ and $w_2\cdot 0^{W-t\cdot 2\wt}$. But as the maximum in  $w_1\cdot 0^{W-t\cdot 2\wt}$ must be $\sigma_{1,t}$ (regardless of the chosen slack), and $\sigma_{1,t}>\sigma_{2,t}\ge\sigma_{2,t+1}\ge\ldots\ge\sigma_{2,\floor{1/2\tau}}$, no single answer is correct for both. Thus, the algorithm must reach a distinct memory configuration for each input, implying a lower bound of $\log_2|L_{\max}|=\log_2{\floor{1/2\tau}+ R\choose R}\ge \floor{1/2\tau}\logp{R/\floor{1/2\tau}}=\Omegap{\oneOverT\log R\tau}$~bits.
\end{proof}
\section{Standard deviation over sliding windows}\label{app:standardDev}
Here, we present our for the standard deviation over a window. 
While algorithms for the sum of a sliding window are known, to the best of our knowledge, no previous solution computes the standard deviation over an (approximate) sliding window.

We denote by $\overline{W}$ the set of items included in the window. 
The algorithm uses a window summing algorithm $\mathbb A$ as a black box.
Here, $\mathbb A$ can be Algorithm~\ref{alg:exact}, Algorithm~\ref{alg:addi}, or Algorithm~\ref{alg:mult}.
We assume that $\mathbb A$ supports two operations:
\begin{enumerate}[I.]
	\item {\sc Update$(x)$} -- process a new element $x\in\frange{R}$.
	\item {\sc Output$()$} -- return a tuple $\langle \widehat{S}, \left|\overline{W}\right|\rangle$ such that $\widehat{S}$ is an estimation of the sum of the last $\left|\overline{W}\right|\approx W$ elements.
\end{enumerate}
%
%
The mean of the window is estimated as $\widehat{m}\triangleq\frac{\widehat{S}}{\left|\overline W\right|}$.

We employ two separate window summing instances (as explained above). The first one simply processes the input and is used to compute the mean. The second one computes the sum of \emph{squared values} over a sliding window. 
This is illustrated in Figure~\ref{fig:stdOverview}.
We use the following identity for the window standard deviation $\sigma_{\overline{W}}$:
\small
$$\sigma_{\overline{W}}=\sqrt{\frac{\sum_{x\in {\overline{W}}}(x-m_{\overline{W}})^2}{\left|\overline{W}\right|-1}}
=\sqrt{\frac{\sum_{x\in {\overline{W}}}x^2-2m_{\overline{W}}\sum_{x\in {\overline{W}}}x+\overline{W}\cdot m_{\overline{W}}^2}{\left|\overline{W}\right|-1}}
=\sqrt{\frac{\sum_{x\in \overline{W}}x^2-\overline{W}\cdot m_{\overline{W}}^2}{\left|\overline{W}\right|-1}}.$$\normalsize
This allow us to compute $\sigma_{\overline{W}}$ from the sum of squares and the mean of the window elements.
A pseudo code of this method appears in Algorithm~\ref{alg:std}.

\begin{algorithm}[]
	\algsize{}
	\caption{Window Standard Deviation Algorithm}\label{alg:std}
	\begin{algorithmic}[1]
		\Statex Initialization: $\mathbb A_x,\mathbb A_{x^2} \triangleq \text{window summing algorithms} $.
		\Function{\add[\inputVariable]}{}
		\State $\mathbb A_x.\mbox{\sc Update}(x)$		
		\State $\mathbb A_{x^2}.\mbox{\sc Update}(x^2)$		
		\EndFunction
		
		\Function{\query}{}
		\State $\langle \widehat{S_x}, \left|\overline{W}\right|\rangle \gets \mathbb A_x.\mbox{\sc Output()}$
		\State $\widehat m \gets \widehat{S_x} / \left|\overline{W}\right|$
		\State $\langle \widehat{S_{x^2}}, \left|\overline{W}\right|\rangle \gets \mathbb A_{x^2}.\mbox{\sc Output()}$
		\State $\widehat{\sigma} \gets \sqrt{\frac{\widehat{S_{x^2}}-\left|\overline{W}\right|\cdot \widehat m^2}{\left|\overline{W}\right|-1}}$
		\State \Return {$\langle\widehat m, \widehat{\sigma}, \left|\overline{W}\right| - W\rangle$}
		\EndFunction
		
	\end{algorithmic}
\end{algorithm}

\begin{figure}[t]
	\centering
	\includegraphics[width=.7\linewidth]{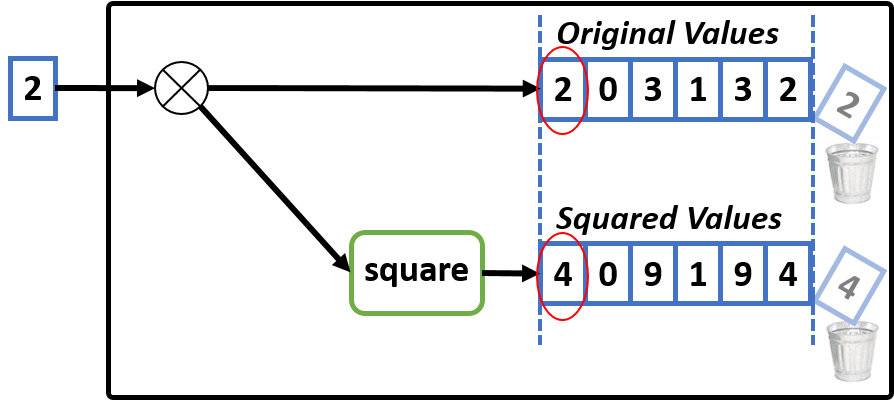}
	\caption{We maintain two summing solutions, one for the original values and one for the squared value. Both algorithm instances track the same window. When estimating the standard deviation of the window, we query both solutions.}
	\label{fig:stdOverview}
\end{figure}

\subsection{Accuracy of standard deviation algorithms}
We now discuss the accuracy that Algorithm~\ref{alg:std} provides, for each specific implementation of the underlying black box $\mathbb A$.
First, if $\mathbb A$ computes the exact sum over the slack window, then all quantities are computed without error.

Next, consider a multiplicative error $\mathbb A$, with a slack window (Algorithm~\ref{alg:mult2}). In this case, the algorithm computes a multiplicative error of the standard deviation (over the window considered by $\mathbb A$). Specifically, if $\mathbb A$ provides a $(1+\epsilon)$ multiplicative approximation of the sum then the standard deviation is estimated within a
multiplicative error of $\sqrt{1+\epsilon}= 1+\epsilon/2 + O(\epsilon^2)$.

Finally, consider an additive $\bserror$ error algorithm on a slack window~(Algorithm~\ref{alg:addi}). In this case, Algorithm~\ref{alg:std} provides an $R\sqrt{\epsilon}$ additive approximation of the window's standard deviation. 

\section{Analysis of allowing slack in count distinct algorithms}\label{app:countDistinct}
\subsection{Background}
Accurately counting distinct elements requires linear space~\cite{HLL}.
Intuitively, one needs to maintain a list of all previously encountered identifiers.
Therefore, accurate measurement does not scale to large streams and approximate solutions are very popular.
Specifically, count distinct algorithms often use randomized estimators~\cite{CD1,CD2,CD3,CD4}.
Randomized algorithms typically use a hash function  $H: \mathbb{D}\to \{0,1\}^{\infty}$ that maps ids to infinite bit strings.
When a maximal cardinality bound is known, finite strings are used and typically 32 bit integers suffice to reach estimations of over $10^9$~\cite{HLL}.
We assume that the hashed values are distributed uniformly at random, i.e., $\forall d\in\mathbb{D}: \Pr[H(d)_i]=0.5$.

Count distinct algorithms look for certain \emph{observables} in the hashes.
For example, some algorithms~\cite{CD1,Giroire2009406} look at the minimal observed hash value as a real number in $[0,1]$ and exploit that
$\mathbb{E}(\min\left(H(\cal{M})\right)) = \frac{1}{n+1}$, where $n$ is the number of the distinct items in the multi-set $\cal{M}$.
Another possibility is to look for patterns of the form $0^{\beta-1}1$~\cite{CD3,HLL}.
When such a pattern is first encountered, it is likely that there were at least $2^{\beta}$ unique elements.

\ifdefined \TWELVEPAGER
\else
Monitoring observables significantly reduces the required amount of space.
The challenge is the variability as a single counter instance is often not accurate enough.
The variability can be reduced by performing $m$ independent experiments.
However, such a simplistic strategy has computational overheads as it requires calculating $m$ different hash functions.
More importantly, such a technique requires a large family of independent hash functions for which no construction is currently known~\cite{Alon:1996:SCA:237814.237823}.
Instead, \emph{stochastic averaging}~\cite{CD4} is used to mimic the effects of multiple experiments with a single hash function.
To do that, we use some of the bits in the hashed value to map each identifier to one of the instances.
That way, each separate instance counts a distinct set of ids and summing the distinct value of all instances yields the number of distinct elements in the stream.
It is shown in~\cite{CD4} that averaging these separate instances yields similar results to $m$ independent experiments.
Specifically, the standard deviation is reduced by a factor of $\frac{1}{\sqrt{m}}$ and enables algorithms to reach high accuracy with minor computational overheads.
\fi
The state of the art count distinct algorithm is~\emph{HyperLogLog (HLL)}~\cite{HLL}, which is being used by Google~\cite{HLLInPractice}.
HLL requires $m$ bytes and its standard deviation is $\sigma \approx \frac{1.04}{\sqrt{m}}$ and was extended to exact windows by~\cite{SlidingHLL,Fusy-HLL}. That extension is used to detect port scans in networked systems~\cite{Chabchoub2014}.
\emph{Exact Window HLL (W-HLL)} requires $5m\ln\left({W/m}\right)$ bytes and its standard deviation is $\sigma \approx \frac{1.04}{\sqrt{m}}$.
In this work, we present \emph{$\tau$-Slack HLL} ($\tau-SHLL$) that requires $(\tau^{-1}+1)m$ bytes and has a standard deviation of $\sigma \approx \frac{1.04}{\sqrt{m}}$.
When $\tau$ is fixed, Slack HLL requires $O(m)$ words. 
When $\tau^{-1}=o(\lnp{W/m})$ it requires asymptotically less space than W-HLL.
%
For completeness, we provide an overview of HLL in Appendix~\ref{app:hll-overview}.

\subsection{$\tau$-Slack HyperLogLog}
We now present the $\tau-SHLL$ algorithm.
We logically divide the stream into fixed $W\tau$ sized blocks.
Our algorithm maintains a cyclic buffer of $\tau^{-1} + 1$ HLL instances.
Each instance has a buffer index in the range $(0,\tau^{-1})$ and the symbol $HLL[i]$ refers to the HLL instance at index $i$, and $HLL[i]_k$ denotes its $k$'th register for $k\in[m-1]$.
We use two counters: \emph{Current Block (CB)} holds values in $\frange{\oneOverT}$ and \emph {Place in Block (PB)} counts how many items are included in the current block.
Initially, CB and PB are set to 0 and new items are always added to $HLL[CB]$.
Each additional item increments $PB$ and once $PB$ reaches the value $W\tau$, we set $PB=0$ and increment $CB$.
We also reset the HLL instance at $HLL[CB]$ by setting all its registers to $-\infty$.
Doing so enables us to forget information that is guaranteed not to be in the window.
To query Slack HLL, we use the maximum of each of the $m$ registers to generate $Z$ and continue as in HLL.
A pseudo code of Slack HLL is found in Algorithm~\ref{alg:slackyHLL} and an example of the algorithm's setup is illustrated in Figure~\ref{fig:CD}. Note that $\alpha_m$ is range correction constant that depends on $m$.
\begin{algorithm}[H]
	\caption{$\tau$-Slack HyperLogLog (for slack windows)}\label{alg:slackyHLL}
	\begin{algorithmic}[1]
		\Statex \textbf{Initialization}: 		$CB  \gets  0, PB  \gets  0$. For $i\in\frange{\oneOverT}$, Initialize $HLL[i]$.
		\Function{\add[\text{}\inputVariable]}{}
		\State \label{line:update}$HLL[CB].Update \left(\inputVariable\right)$				
		\State $PB  \gets  (PB + 1) \mod \wt$
		\If {$PB  \gets  0$} \Comment{End of block}
		\State$CB \gets  (CB+1) \mod (\oneOverT+1)$
		\State Initialize $HLL[CB]$
		\EndIf
		\EndFunction
		
		\Function{{\sc Query}(\text{}$\inputVariable$)}{}
		\State For every $k\in\xrange{m}$, compute $M_{max}[k]\triangleq\max\set{HLL[i]_k\mid i\in[\oneOverT+1]}$.		
		\State $Z  \gets  {\left( {\sum\nolimits_{j  \gets  0}^{m - 1} {{2^{ - M_{max}\left[ j \right]}}} } \right)^{ - 1}}$
		\State \textbf{return} $\langle\alpha_m \cdot m^{2}\cdot Z, PB\rangle$
		\EndFunction
	\end{algorithmic}
\end{algorithm}
%
\begin{figure}[]
	
	\center{
		\includegraphics[width = \linewidth]{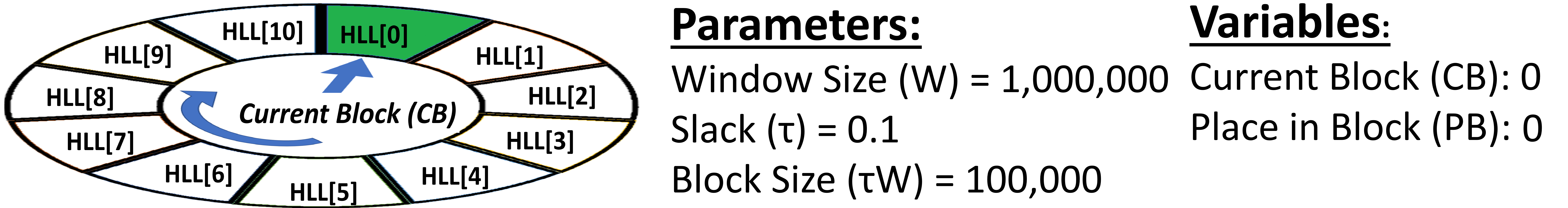}
		\caption{\label{fig:CD} An example of the $\tau$-Slack HLL algorithm with a window of size $W$ and $\tau =0.1$. $\tau-SHLL$ relies on $\oneOverT+1=11$ HLL instances, each is used to count 100,000 events. Every $\wt$ events we increment the CB counter and reset HLL[CB]. This enables our algorithm to work here on window sizes (including slack) of minimum 1,000,000 events and maximum 1,100,000~events.
	}}	\vspace*{-0.5cm}
\end{figure}
\noindent The next theorem, whose proof is deferred to Appendix~\ref{app:shll-proof}, shows that $\tau-SHLL$ is correct.
\begin{theorem}
	\label{thm:CDCorrect}
	Let $\langle\widehat{D},PB\rangle$ be the query result of Slack HLL, then $\widehat{D}$ is the query result of $HLL$ for a stream containing the last $W+PB$ events (for streams longer than $W$).
\end{theorem}

\section{Hyper LogLog (HLL) Overview}\label{app:hll-overview}

We provide a brief overview of the HLL algorithm~\cite{HLL}.
HLL uses a hash function that maps each identifier to an infinite string of $0$,$1$ bits; $H: \mathbb{D} \to \{0,1\}^{\infty}$, where $\mathbb{D}$ is the identifier domain.
Given $s \in \{0,1\}^{\infty}$, the operator $lsb(s)$ returns the position of the leftmost $1$-bit, e.g., $lsb(0001...) = 3$ (counting from 0).
Intuitively, a string with $k$ leading zeros is expected to appear after $2^{k}$ events.
That is, $HLL$ stores the largest previously encountered $lsb$ value in a counter.
The space complexity is $\log \log \left(|\cal{M}|\right)$, which in practice means that a single byte suffices for most scenarios.

To augment precision, $m$ different $HLL$ counters are used.
For simplicity, we assume that $m=2^{b}$ for some positive $b\in \mathbb{N}$.
Stochastic averaging is performed by using the first $b$ bits of each hashed value to determine the $HLL$ instance to be updated.

To satisfy a query, we read all $m$ estimations, calculate their harmonic average and normalize the result.
The technical details of how to best interpret the result can be found in~\cite{HLL,HLLInPractice}.
Algorithm~\ref{alg:HLL} provides pseudo code of the HLL algorithm.
As can be observed, stochastic averaging is performed in Line~\ref{line:stochastic}.
The notation $\left(h_0, h_1,...h_{b-1}\right)$ describes the bit composition of $h$.
The query algorithms uses the harmonic average of the different experiments $ \left(M[0],M[1],...,M[m]\right)$ and the result is then normalized with a constant that depends on $m$ (specifically, $\alpha_m \cdot m^{2}$).
\begin{algorithm}[]
	\caption{HyperLogLog (for streams)}\label{alg:HLL}
	\begin{algorithmic}[1]
		\State Initialization: $M[0] = M[1] = M[m-1] =  -\infty$.
		\Function{\add[\text{element }\inputVariable]}{}
		\State $h = H\left(\inputVariable\right)$
		\State $id = \left(h_0, h_1,...h_{b-1}\right)$ \Comment{We use the first $b=\log m$ bits to determin which register to update}
		\State $w = \left(h_b, x_{b+1},...\right)$ \Comment{We compute the number of leading zeros $\rho(w)$ in the following bits.} \label{line:stochastic}
		\State $M[id] = max\left(M[id],\rho(w)\right)$	
		\EndFunction
		
		\Function{query[\text{element }\inputVariable]}{}
		\State $Z = {\left( {\sum\nolimits_{j = 0}^{m - 1} {{2^{ - M\left[ j \right]}}} } \right)^{ - 1}}$
		\State \textbf{return} $\alpha_m \cdot m^{2}\cdot Z$.
		\EndFunction
	\end{algorithmic}
\end{algorithm}

\section{Proof of Theorem~\ref{thm:CDCorrect}}\label{app:shll-proof}
Before we prove Theorem~\ref{thm:CDCorrect}, we need to prove an auxiliary lemma. 
\begin{lemma}\vspace*{-0.3cm}
	\label{lemma:working}
	The last $\min (|\cal{M}|$$,W+PB)$ updates are performed on all $HLL$ instances.
\end{lemma}
\begin{proof}
	Initially, when $|\cal{M}|<W+W\tau$, every $HLL$ instance that is initialized is already empty.
    Thus, since every Slack HLL updates an $HLL$ instance (in Line~\ref{line:update}), the number of update operations is $|\cal{M}|$.
	When $|\cal{M}| =W+W\tau$, $HLL[0]$ is initialized at first and we lose the oldest $\tau W$ events.
    At that point, the number of events is $W$ and $PB=0$.
    In each subsequent update, both $PB$ and the number of events are increased by $1$ until $PB=\tau W$ again.
    At that point, $PB$ is set back to $0$ and the number of events drops to the last $W$.
\end{proof}

We are now ready to prove Theorem~\ref{thm:CDCorrect}. 
\begin{proof}
	Note that $W'=W+PB$ and therefore Lemma~\ref{lemma:working} guarantees that the last $W'$ events are summarized in Slack HLL.
    Therefore, $M_{max}[i]$ used to generate $Z$ in Slack HLL has the highest $\rho$ value of all the last $W'$ elements.
    Consequently, the value that determined $M_{max}[i]$ in Slack HLL also determines $M[i]$ in an HLL that summarizes the last $W$ elements.
\end{proof}

\end{document}